\documentclass[12pt,journal,final,onecolumn]{IEEEtran}

\makeatletter

\let\proof\@undefined
\let\endproof\@undefined
\makeatother

\usepackage[dvips]{graphicx}
\usepackage{epsfig}
\usepackage{amsmath}
\usepackage{amssymb}
\usepackage{amsthm}
\usepackage{amsfonts}
\usepackage{bm}
\usepackage{color}
\usepackage[mathscr]{eucal}
\usepackage{cite}

\graphicspath{{figs/}}

\newcommand{\mc}[1]{\mathcal{#1}}

\newcommand{\mbb}[1]{\mathbb{#1}}

\newcommand{\tsf}[1]{\textsf{#1}}

\newcommand{\defeq}{\triangleq}

\newcommand{\R}{\mathbb{R}}
\newcommand{\Rp}{\R_{+}}

\newcommand{\Pp}{\mathbb{P}}
\newcommand{\ind}{1\hspace{-2.5mm}{1}}
\newcommand{\E}{\mathbb{E}}

\newcommand{\ceil}[1]{\left\lceil{#1}\right\rceil}
\newcommand{\floor}[1]{\left\lfloor{#1}\right\rfloor}
\newcommand{\abs}[1]{\left\lvert{#1}\right\rvert}
\newcommand{\card}[1]{\abs{#1}}

\newcommand{\lambdauc}{\lambda^{\textup{\tsf{UC}}}}
\newcommand{\Lambdauc}{\Lambda^{\textup{\tsf{UC}}}}
\newcommand{\lambdamc}{\lambda^{\textup{\tsf{MC}}}}
\newcommand{\Lambdamc}{\Lambda^{\textup{\tsf{MC}}}}

\newcommand{\hLambdauc}{\widehat{\Lambda}^{\textup{\tsf{UC}}}}

\newcommand{\hLambdamc}{\widehat{\Lambda}^{\textup{\tsf{MC}}}}

\newcommand{\Lambdabuc}{\Lambda^{\textup{\tsf{BUC}}}}

\newcommand{\Lambdabmc}{\Lambda^{\textup{\tsf{BMC}}}}

\newcommand{\hLambdabuc}{\widehat{\Lambda}^{\textup{\tsf{BUC}}}}

\newcommand{\hLambdabmc}{\widehat{\Lambda}^{\textup{\tsf{BMC}}}}
\newcommand{\Buc}{\mc{B}^{\textup{\tsf{UC}}}}
\newcommand{\Bmc}{\mc{B}^{\textup{\tsf{MC}}}}

\newcommand{\ald}{\bar{\alpha}}

\newtheorem{lemma}{Lemma}

\newtheorem{theorem}[lemma]{Theorem}

\theoremstyle{definition}
\newtheorem{egdummy}{Example}
\newenvironment{example}{%
    \begin{egdummy}%
        \upshape%
        
}
{\qed%
\end{egdummy}}

\theoremstyle{remark}

\begin{document}

\bibliographystyle{unsrt}

\title{The Balanced Unicast and Multicast Capacity Regions of Large Wireless Networks} 

\author{Urs Niesen, Piyush Gupta, and Devavrat Shah
\thanks{This work was supported, in parts, by DARPA grant
(ITMANET) 18870740-37362-C, by NSF grants CCR-0325673 and CNS-0519535, and by a AFOSR grant under the complex networks program. The material in this paper
was presented, in parts, at the Allerton Conference on Communication,
Control, and Computing, Monticello, IL, September 2008, and at the IEEE
INFOCOM Conference, Rio de Janeiro, Brazil, April 2009.}
\thanks{U. Niesen was with the Laboratory for Information and Decision Systems 
at the Massachusetts Institute of Technology. He is now with the
Mathematics of Networks and Communications Research Department, Bell
Labs, Alcatel-Lucent.
Email: urs.niesen@alcatel-lucent.com}
\thanks{P. Gupta is with the Mathematics of Networks and 
Communications Research Department, Bell Labs, Alcatel-Lucent. 
Email: pgupta@research.bell-labs.com}
\thanks{D. Shah is with the Laboratory for Information and Decision Systems
at the Massachusetts Institute of Technology. 
Email: devavrat@mit.edu}
}

\maketitle

\begin{abstract}
    We consider the question of determining the scaling of the
    $n^2$-dimensional balanced unicast and the $n 2^n$-dimensional
    balanced multicast capacity regions of a wireless network with $n$
    nodes placed uniformly at random in a square region of area $n$ and
    communicating over Gaussian fading channels. We identify this scaling
    of both the balanced unicast and multicast capacity regions in terms of
    $\Theta(n)$, out of $2^n$ total possible, cuts. These cuts only
    depend on the geometry of the locations of the source nodes and
    their destination nodes and the traffic demands between them, and
    thus can be readily evaluated. Our results are constructive and
    provide optimal (in the scaling sense) communication schemes. 
\end{abstract}

\section{Introduction}
\label{sec:intro}

Characterizing the capacity region of wireless networks is a long
standing open problem in information theory.  The exact capacity region
is, in fact, not known for even simple networks like a three node relay
channel or a four node interference channel. In this paper, we consider
the question of approximately determining the unicast and multicast
capacity regions of wireless networks by identifying their scaling in
terms of the number $n$ of nodes in the network.

\subsection{Related Work} 

In the last decade, exciting progress has been made towards
approximating the capacity region of wireless networks. We shall
mention a small subset of work related to this paper. 

We first consider unicast traffic. The unicast capacity region of a
wireless network with $n$ nodes is the set of all simultaneously
achievable rates between all possible $n^2$ source-destination pairs.
Since finding this unicast capacity region of a wireless network exactly
seems to be intractable, Gupta and Kumar proposed a simpler but
insightful question in \cite{gup}. First, instead of asking for the
entire $n^2$-dimensional unicast capacity region of a wireless network
with $n$ nodes, attention was restricted to the scenario where each node
is source exactly once and chooses its destination uniformly at random
from among all the other nodes.  All these $n$ source-destination pairs
communicate at the same rate, and the interest is in finding the maximal
achievable such rate.  Second, instead of insisting on finding this
maximal rate exactly, they focused on its asymptotic behavior as the
number of nodes $n$ grows to infinity.

This setup has indeed turned out to be more amenable to analysis. In
\cite{gup}, it was shown that under random placement of nodes in a given
region and under certain models of communication motivated by current
technology (called \emph{protocol channel model} in the following), the
per-node rate for random source-destination pairing with uniform traffic
can scale at most as $O(n^{-1/2})$ and this can be achieved (within
poly-logarithmic factor in $n$) by a simple scheme based on multi-hop
communication. Many works since then have broadened the channel and
communication models under which similar results can be proved (see, for
example, \cite{xie,jov,xue,lev,xie2,fra,gup2,xie3,kra,aer,ozg,nie}).  In
particular, under the \emph{Gaussian fading channel model} with a
power-loss of $r^{-\alpha}$ for signals sent over a distance of $r$, it
was shown in \cite{ozg} that in extended wireless networks (i.e., $n$
nodes are located in a region of area $\Theta(n)$) the largest uniformly
achievable per-node rate under random source-destination pairing scales
essentially like $\Theta\big(n^{1-\min\{3,\alpha\}/2}\big)$. 

Analyzing such random source-destination pairing with uniform traffic
yields information about the $n^2$-dimensional unicast capacity region
along one dimension. Hence, the results in \cite{gup} and in \cite{ozg}
mentioned above provide a complete characterization of the scaling of
this one-dimensional slice of the capacity region for the protocol and
Gaussian fading channel models, respectively. It is therefore natural to
ask if the scaling of the entire $n^2$-dimensional unicast capacity
region can be characterized. To this end, we describe two related
approaches taken in recent works. 

One approach, taken by Madan, Shah, and L\'ev\^eque \cite{mad}, builds
upon the celebrated works of Leighton and Rao \cite{LR98} and Linial,
London, and Rabinovich \cite{llr} on the approximate characterization of
the unicast capacity region of capacitated wireline networks. For such
wireline networks, the scaling of the unicast capacity region is
determined (within a $\log(n)$ factor) by the minimum weighted cut of
the network graph. As shown in \cite{mad}, this naturally extends to
wireless networks under the protocol channel model, providing an
approximation of the unicast capacity region in this case. 

Another approach, first introduced by Gupta and Kumar \cite{gup},
utilizes geometric properties of the wireless network. Specifically, the
notion of the {\em transport capacity} of a network, which is the
rate-distance product summed over all source-destination pairs, was
introduced in \cite{gup}. It was shown that in an extended wireless
network with $n$ nodes and under the protocol channel model, the
transport capacity can scale at most as $\Theta(n)$. This bound on the
transport capacity provides a hyper-plane which has the capacity region
and origin on the same side. Through a repeated application of this
transport capacity bound at different scales \cite{sub, sub2} obtained
an implicit characterization of the unicast capacity region under the
protocol channel model. 

For the Gaussian fading channel model, asymptotic upper bounds for the
transport capacity were obtained in \cite{xie, jov}, and for more
general distance weighted sum rates in \cite{ahm}. 

So far, we have only considered unicast traffic. We now turn to
multicast traffic. The multicast capacity region of a wireless network
with $n$ nodes is the set of all simultaneously achievable rates between
all possible $n 2^n$ source--multicast-group pairs. Instead of
considering this multicast capacity region directly, various authors
have analyzed the scaling of restricted traffic patterns under a
protocol channel model assumption (see
\cite{ltf,sls,keriri06,chsa05,jaro05}, among others). For example, in
\cite{ltf}, Li, Tang, and Frieder obtained a scaling characterization
under a protocol channel model and random node placement for multicast
traffic when each node chooses a certain number of its destinations
uniformly at random. Independently, in \cite{sls}, Shakkottai, Liu, and
Srikant considered a similar setup and also obtained the precise scaling
when sources and their multicast destinations are chosen at random. Both
of these results are non information-theoretic (in that they assume a
protocol channel model). Furthermore, they provide information about the
scaling of the $n 2^n$-dimensional multicast capacity region only
along one particular dimension.

\subsection{Our Contributions}

Despite the long list of results, the question of approximately
characterizing the unicast capacity region under the Gaussian fading
channel model remains far from being resolved. In fact, for Gaussian
fading channels, the only traffic pattern that is well understood is
random source-destination pairing with uniform rate. This is limiting in
several aspects. First, by choosing for each source a destination at
random, most source-destination pairs will be at a distance of the
diameter of the network with high probability, i.e., at distance
$\Theta(\sqrt{n})$ for an extended network. However, in many wireless
networks some degree of locality of source-destination pairs can be
expected. Second, all source-destination pairs are assumed to be
communicating at uniform rate. Again, in many settings we would expect
nodes to be generating traffic at widely varying rates. Third, each node
is source exactly once, and destination on average once. However, in
many scenarios the same source node (e.g., a server) might transmit data
to many different destination nodes, or the same destination node might
request data from many different source nodes. All these heterogeneities
in the traffic demands can result in different scaling behavior of the
performance of the wireless network than what is obtained for random
source-destination pairing with uniform rate. 

As is pointed out in the last section, even less is known about the
multicast capacity region under Gaussian fading. In fact, the only
available results are for the protocol channel model, and even there
only for special traffic patterns resulting from randomly choosing
sources and their multicast groups and assuming uniform rate.  To the
best of our knowledge, no information-theoretic results (i.e., assuming
Gaussian fading channels) are available even for special traffic
patterns.

We address these issues by analyzing the scaling of a broad class of
traffic, termed \emph{balanced traffic} in the following, in a wireless
network of $n$ randomly placed nodes under a Gaussian fading channel
model. The notion of balanced traffic is a natural generalization of
symmetric traffic, in which the data to be transmitted from a node $u$
to a node $v$ is equal to the amount of data to be transmitted from $v$ to
$u$. We analyze the scaling of the set of achievable balanced unicast
traffic (the \emph{balanced unicast capacity region}) and achievable
balanced multicast traffic (the \emph{balanced multicast capacity
region}). The balanced unicast capacity region provides information
about $n^2-n$ of the $n^2$ dimensions of the unicast capacity region;
the balanced multicast capacity region provides information about
$n2^n-n$ of the $n2^n$ dimensions of the multicast capacity region.

As a first set of results of this paper, we present an approximate
characterization of the balanced unicast and multicast capacity regions.
We show that both regions can be approximated by a polytope with less
than $2n$ faces, each corresponding to a distinct cut (i.e., a subset of
nodes) in the wireless network. This polyhedral characterization
provides a succinct approximate description of the balanced unicast and
multicast capacity regions even for large values of $n$. Moreover, it
shows that only $2n$ out of $2^n$ possible cuts in the wireless network
are asymptotically relevant and reveals the geometric structure of these
relevant cuts.

Second, we establish the approximate equivalence of the wireless network
and a wireline tree graph, in the sense that balanced traffic can be
transmitted reliably over the wireless network if and only if
approximately the same traffic can be routed over the tree graph. This
equivalence is the key component in the derivation of the approximation
result for the balanced unicast and multicast capacity regions and
provides insight into the structure of large wireless networks.

Third, we propose a novel three-layer communication architecture that
achieves (in the scaling sense) the entire balanced unicast and multicast
capacity regions. The top layer of this scheme treats the wireless
network as the aforementioned tree graph and routes messages between
sources and their destinations---dealing with heterogeneous traffic
demands. The middle layer of this scheme provides this tree abstraction
to the top layer by appropriately distributing and concentrating traffic
over the wireless network---choosing the level of cooperation in the
network. The bottom layer implements this distribution and concentration
of messages in the wireless network---dealing with interference and
noise. The approximate optimality of this three-layer architecture
implies that a separation based approach, in which routing is performed
independently of the physical layer, is order-optimal. In other words,
techniques such as network coding can provide at most a small (in the
scaling sense) multiplicative gain for transmission of balanced unicast
or multicast traffic in wireless networks.

\subsection{Organization}

The remainder of this paper is organized as follows. Section
\ref{sec:model} introduces the network model and notation. Section
\ref{sec:main} presents our main results. We illustrate the strength of
these results in Section \ref{sec:examples} by analyzing various example
scenarios with heterogeneous unicast and multicast traffic patterns.
Section \ref{sec:schemes} provides a high level description of the
proposed communication schemes. Sections
\ref{sec:aux}-\ref{sec:proof_multicast} contain proofs. Finally,
Sections \ref{sec:discussion} and \ref{sec:conclusions} contain
discussions and concluding remarks.

\section{Models and Notation}
\label{sec:model}

In this section, we discuss network and traffic models, and we
introduce some notational conventions.

\subsection{Network Model}

Consider the square region 
\begin{equation*}
    A(n) \defeq [0,\sqrt{n}]^2
\end{equation*}
and let $V(n)\subset A(n)$ be a set of $\abs{V(n)} = n$ nodes on $A(n)$.
Each such node represents a wireless device, and the $n$ nodes together
form a wireless network. This setting with $n$ nodes on a square of area
$n$ is referred to as an \emph{extended network}. Throughout this paper,
we consider this extended network setting. However, all results carry
over for \emph{dense networks}, where $n$ nodes are placed on a square
of unit area (see Section \ref{sec:dense} for the details). 

We use the same channel model as in~\cite{ozg}.
Namely, the received signal at node $v$ and time
$t$ is
\begin{equation*}
    y_v[t] \defeq \sum_{u\in V(n)\setminus\{v\}}h_{u,v}[t]x_u[t]+z_v[t]
\end{equation*}
for all $v\in V(n), t\in\mbb{N}$, where the $\{x_u[t]\}_{u,t}$ are the
signals sent by the nodes in $V(n)$.  We impose an average
power constraint of $1$ on the signal $\{x_u[t]\}_{t}$ for every node
$u\in V(n)$.  The additive noise terms $\{z_v[t]\}_{v,t}$ are
independent and identically distributed (i.i.d.) circularly symmetric
complex Gaussian random variables with mean $0$ and variance $1$, and
\begin{equation*}
    h_{u,v}[t] \defeq r_{u,v}^{-\alpha/2}\exp(\sqrt{-1}\theta_{u,v}[t]),
\end{equation*}
for \emph{path-loss exponent} $\alpha>2$, and where $r_{u,v}$ is the
Euclidean distance between $u$ and $v$.  As a function of $u,v\in V(n)$,
we assume that $\{\theta_{u,v}[t]\}_{u,v}$ are i.i.d.\footnote{It is
worth pointing out that recent results \cite{fra2} suggest that under
certain assumptions on scattering elements, for $\alpha\in(2,3)$ and
very large values of $n$, the i.i.d. phase assumption does not
accurately reflect the physical behavior of the wireless channel.
However, in follow-up work \cite{fra3} the authors show that under
different assumptions on the scatterers this assumption is still
justified in the $\alpha\in(2,3)$ regime even for very large values of
$n$. This indicates that the issue of channel modeling for large
networks in the low path-loss regime is somewhat delicate and requires
further investigation.} with uniform distribution on $[0,2\pi)$. As a
function of $t$, we either assume that $\{\theta_{u,v}[t]\}_{t}$ is
stationary and ergodic, which is called \emph{fast fading} in the
following, or we assume $\{\theta_{u,v}[t]\}_{t}$ is constant, which is
called \emph{slow fading} in the following. In either case, we assume
full channel state information (CSI) is available at all nodes, i.e.,
each node knows all $\{h_{u,v}[t]\}_{u,v}$ at time $t$. This full CSI
assumption is rather strong, and so it is worth commenting on. All the
converse results presented are proved under the full CSI assumption and
are hence also valid under more realistic assumptions on the
availability of CSI. Moreover, it can be shown that for achievability
only $2$-bit quantized CSI is necessary for path-loss exponent
$\alpha\in(2,3]$ and no CSI is necessary for $\alpha>3$ to achieve the
same scaling behavior.

\subsection{Traffic Model}

A \emph{unicast traffic matrix} $\lambdauc\in\Rp^{n\times n}$ associates
with each pair $u,w\in V(n)$ the rate $\lambdauc_{u,w}$ at which node $u$
wants to communicate to node $w$. We assume that messages for distinct
source-destination pairs $(u,w)$ are independent.  However, we allow the
same node $u$ to be source for multiple destinations, and the same node
$w$ to be destination for multiple sources. In other words, we consider
general unicast traffic.  The \emph{unicast capacity region}
$\Lambdauc(n)\subset\Rp^{n\times n}$ of the wireless network is the
collection of achievable unicast traffic matrices, i.e.,
$\lambdauc\in\Lambdauc(n)$ if and only if every source-destination pair
$(u,w)\in V(n)\times V(n)$ can reliably communicate independent messages
at rate $\lambdauc_{u,w}$. 

A \emph{multicast traffic matrix} $\lambdamc\in\Rp^{n\times 2^n}$
associates with each pair $u\in V(n), W\subset V(n)$ the rate
$\lambdamc_{u,W}$ at which node $u$ wants to multicast a message to the
nodes in  $W$. In other words, all nodes in $W$ want to obtain the same
message from $u$.  We assume that messages for distinct
source--multicast-group pairs $(u,W)$ are independent. However, we allow
the same node $u$ to be source for several multicast-groups, and the
same set $W$ of nodes to be multicast destination for multiple sources.
In other words, we consider general multicast traffic.  The
\emph{multicast capacity region} $\Lambdamc(n)\subset\Rp^{n\times 2^n}$
is the collection of achievable multicast traffic matrices, i.e,.
$\lambdamc\in\Lambdamc(n)$ if and only if every
source--multicast-group pair $(u,W)$ can reliably communicate
independent messages at rate $\lambdamc_{u,W}$.

The following example illustrates the concept of unicast and multicast
traffic matrices.
\begin{example}
    Assume $n=4$, and label the nodes as $\{u_i\}_{i=1}^4= V(n)$.
    Assume further node $u_1$ needs to transmit a message $m_{1,2}$ to
    node $u_2$ at rate $1$ bit per channel use, and an independent
    message $m_{1,3}$ to node $u_3$ at rate $2$ bits per channel use.
    Node $u_2$ needs to transmit a message $m_{2,3}$ to node $u_3$ at
    rate $4$ bits per channel use. All the messages $m_{1,2}, m_{1,3},
    m_{2,3}$ are independent. This traffic pattern can be described by a
    unicast traffic matrix $\lambdauc\in\Rp^{4\times 4}$ with
    $\lambdauc_{u_1,u_2}=1$, $\lambdauc_{u_1,u_3}=2$,
    $\lambdauc_{u_2,u_3}=4$, and $\lambdauc_{u,w}=0$ otherwise. Note that
    in this example node $u_1$ is source for two (independent) messages,
    and node $u_3$ is destination for two (again independent) messages.
    Node $u_4$ in this example is neither source nor destination for any
    message and can be understood as a helper node. 
    
    Assume now that node $u_1$ needs to transmit the same message
    $m_{1,\{2,3,4\}}$ to all nodes $u_2,u_3,u_4$ at a rate of $1$ bit
    per channel use, and an independent message $m_{1,\{2\}}$ to only
    node $2$ at rate $2$ bits per channel use. Node $2$ needs to
    transmit a message $m_{2,\{1,3\}}$ to both $u_1,u_3$ at rate $4$
    bits per channel use. All the messages $m_{1,\{2,3,4\}},
    m_{1,\{2\}}, m_{2,\{1,3\}}$ are independent. This traffic pattern
    can be described by a multicast traffic matrix
    $\lambdamc\in\Rp^{4\times 16}$ with
    $\lambdamc_{u_1,\{u_2,u_3,u_4\}}=1$, $\lambdamc_{u_1,\{u_2\}}=2$,
    $\lambdamc_{u_2,\{u_1,u_3\}}=4$, and $\lambdamc_{u,W}=0$ otherwise. 
    Note that in this example node $u_1$ is source for two (independent)
    multicast messages, and node $2$ and $3$ are destinations for more
    than one message. The message $m_{1,\{2,3,4\}}$ is destined for 
    all the nodes in the network and can hence be understood as a
    broadcast message. The message $m_{1,\{2\}}$ is only
    destined for one node and can hence be understood as a private
    message.
\end{example}

In the following, we will be interested in \emph{balanced} traffic
matrices that satisfy certain symmetry properties. Consider a symmetric
unicast traffic matrix $\lambdauc$ satisfying
$\lambdauc_{u,w}=\lambdauc_{w,u}$ for all node pairs $u,w\in V(n)$.
The notion of a balanced traffic matrix generalizes this idea of
symmetric traffic. 

Before we provide a precise definition of balanced traffic, we need to
introduce some notation. Partition $A(n)$ into several square-grids.
The $\ell$-th square-grid divides $A(n)$ into $4^\ell$ squares, each of
sidelength $2^{-\ell}\sqrt{n}$, denoted by
$\{A_{\ell,i}(n)\}_{i=1}^{4^{\ell}}$.  Let $V_{\ell,i}(n) \subset V(n)$
be the nodes in $A_{\ell,i}(n)$ (see Figure \ref{fig:grid}).  The square
grids in levels $\ell\in\{1,\ldots,L(n)\}$ with\footnote{All logarithms
are with respect to base $2$.}
\begin{equation*}
    L(n) \defeq \frac{1}{2}\log(n)\big(1-\log^{-1/2}(n)\big),
\end{equation*}
will be of particular importance. Note that $L(n)$ is chosen such that
\begin{equation*}
    4^{-L(n)}n = n^{\log^{-1/2}(n)},
\end{equation*}
and hence
\begin{equation*}
    \lim_{n\to\infty}\card{A_{L(n),i}(n)}
    = \lim_{n\to\infty}4^{-L(n)}n 
    = \infty.
\end{equation*}
while at the same time
\begin{equation*}
    \card{A_{L(n),i}(n)}
    = 4^{-L(n)}n 
    \leq n^{o(1)},
\end{equation*}
as $n\to\infty$. In other words, the area of the region $A_{L(n),i}(n)$
at level $\ell=L(n)$ grows to infinity as $n\to\infty$, but much slower
than $n$.
\begin{figure}[!ht]
    \begin{center}
        \input{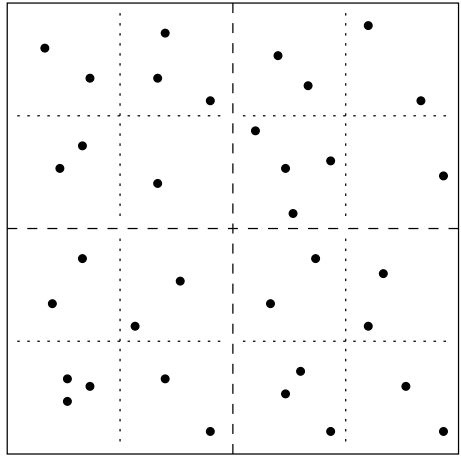}
    \end{center}

    \caption{Square-grids with $0\leq\ell\leq 2$.  The grid at level
    $\ell=0$ is the area $A(n)$ itself. The grid at level $\ell=1$ is
    indicated by the dashed lines. The grid at level $\ell=2$ by the
    dashed and the dotted lines. Assume for the sake of example that the
    subsquares are numbered from left to right and then from bottom to
    top (the precise order of numbering is immaterial).  Then
    $V_{0,1}(n)$ are all the nodes $V(n)$, $V_{1,1}(n)$ are the nine
    nodes in the lower left corner (separated by dashed lines), and
    $V_{2,1}(n)$ are the three nodes in the lower left corner (separated
    by dotted lines).}

    \label{fig:grid}
\end{figure}

A unicast traffic matrix $\lambdauc$ is \emph{$\gamma$-balanced} if
\begin{equation}
    \label{eq:unicast_balanced}
    \sum_{u\notin V_{\ell,i}(n)}\sum_{w\in V_{\ell,i}(n)}\lambdauc_{u,w}
    \leq \gamma \sum_{u\in V_{\ell,i}(n)}\sum_{w\notin V_{\ell,i}(n)}\lambdauc_{u,w},
\end{equation}
for all $\ell\in\{1,\ldots,L(n)\}$ and $i\in\{1,\ldots 4^\ell\}$. In
other words, for a balanced unicast traffic matrix the amount of traffic to the
nodes $V_{\ell,i}(n)$ is not much larger than the amount of traffic from  
them. In particular, all symmetric traffic matrices, i.e.,
satisfying $\lambdauc_{u,w}=\lambdauc_{w,u}$, are $1$-balanced.  Denote
by $\Buc(n)\subset \Rp^{n\times n}$ the collection of all
$\gamma(n)$-balanced unicast traffic matrices for some fixed
$\gamma(n)=n^{o(1)}$.  In the following, we refer to traffic matrices
$\lambdauc\in\Buc(n)$ simply as balanced traffic matrices.  The
\emph{balanced unicast capacity region}
$\Lambdabuc(n)\subset\Rp^{n\times n}$  of the wireless network is the
collection of balanced unicast traffic matrices that are achievable,
i.e., 
\begin{equation*}
    \Lambdabuc(n)\defeq \Lambdauc(n)\cap \Buc(n).
\end{equation*}
Note that \eqref{eq:unicast_balanced} imposes at most $n$ linear
inequality constraints, and hence  $\Lambdauc(n)$ and $\Lambdabuc(n)$
coincide along at least $n^2-n$ of $n^2$ total dimensions.

A multicast traffic matrix $\lambdamc$ is \emph{$\gamma$-balanced} if
\begin{equation}
    \label{eq:multicast_balanced}
    \sum_{u\notin V_{\ell,i}(n)}
    \sum_{\substack{W\subset V(n):
    \\ W\cap V_{\ell,i}(n)\neq\emptyset}}\lambdamc_{u,W}
    \leq \gamma \sum_{u\in V_{\ell,i}(n)}
    \sum_{\substack{W\subset V(n):
    \\ W\setminus V_{\ell,i}(n)\neq\emptyset}}\lambdamc_{u,W}
\end{equation}
for all $\ell\in\{1,\ldots,L(n)\}$, $i\in\{1,\ldots 4^\ell\}$. Thus, for
$\gamma$-balanced multicast traffic, the amount of traffic to the nodes
$V_{\ell,i}(n)$ is not much larger than the amount of traffic from them.
This is the natural generalization of the notion of $\gamma$-balanced
unicast traffic to the multicast case. Denote by
$\Bmc(n)\subset\Rp^{n\times 2^n}$ the collection of all
$\gamma(n)$-balanced multicast traffic matrices for some fixed
$\gamma(n)=n^{o(1)}$.  As before, we will refer to a multicast traffic
matrix $\lambdamc\in\Bmc(n)$ simply as balanced multicast traffic
matrix.  The \emph{balanced multicast capacity region}
$\Lambdabmc(n)\subset\Rp^{n\times 2^n}$ of the wireless network is the
collection of balanced multicast traffic matrices that are achievable,
i.e., 
\begin{equation*}
    \Lambdabmc(n)\defeq \Lambdamc(n)\cap \Bmc(n).
\end{equation*}
Equation \eqref{eq:multicast_balanced} imposes at most $n$ linear
inequality constraints, and hence $\Lambdamc(n)$ and $\Lambdabmc(n)$
coincide along at least $n2^n-n$ of $n2^n$ total dimensions.

\subsection{Notational Conventions}

Throughout, $\{K_i\}_i$, $K$, $\widetilde{K}$, \ldots, indicate strictly
positive finite constants independent of $n$ and $\ell$. To simplify
notation, we assume, when necessary, that large real numbers are
integers and omit $\ceil{\cdot}$ and $\floor{\cdot}$ operators. For the
same reason, we also suppress dependence on $n$ within proofs whenever
this dependence is clear from the context.

\section{Main Results}
\label{sec:main}

In this section, we present the main results of this paper. In Section
\ref{sec:unicast}, we provide an approximate (i.e., scaling)
characterization of the entire balanced unicast capacity region
$\Lambdabuc(n)$ of the wireless network, and in Section
\ref{sec:multicast}, we provide a scaling characterization of the entire
balanced multicast capacity region $\Lambdabmc(n)$.  In Section
\ref{sec:implications}, we discuss implications of these results on the
behavior of the unicast and multicast capacity regions for large values
of $n$. In Section \ref{sec:computation}, we consider computational
aspects.

\subsection{Balanced Unicast Capacity Region}
\label{sec:unicast}

Here we present a scaling characterization of the complete balanced unicast
capacity region $\Lambdabuc(n)$.  

Define
\begin{equation}
    \label{eq:approx_unicast}
    \begin{aligned}
        \hLambdauc(n)
        \defeq \Big\{\lambdauc\in\Rp^{n\times n}:
        & \sum_{u\in V_{\ell,i}(n)}\sum_{w\notin V_{\ell,i}(n)}\lambdauc_{u,w}
        \leq (4^{-\ell}n)^{2-\min\{3,\alpha\}/2} \\
        & \qquad\qquad \forall \ell\in\{1,\ldots,L(n)\}, i\in\{1,\ldots, 4^\ell\}, \\
        & \sum_{w \neq u}(\lambdauc_{u,w}+\lambdauc_{w,u}) \leq 1
        \ \forall u\in V(n)
        \Big\},
    \end{aligned}
\end{equation}
and set
\begin{equation*}
    \hLambdabuc(n) \defeq \hLambdauc(n)\cap\Buc(n).
\end{equation*}
$\hLambdabuc(n)$ is the collection of all balanced unicast traffic matrices
$\lambdauc$ such that for various cuts $S\subset V(n)$ in the network, the
total traffic demand (in either one or both directions)
\begin{align*}
    & \sum_{u\in S}\sum_{w\notin S}\lambdauc_{u,w}, \\
    & \sum_{u\in S}\sum_{w\notin S}(\lambdauc_{u,w}+\lambdauc_{w,u}),
\end{align*}
across the cut $S$ is not too big. Note that the number of cuts $S$ we
need to consider is actually quite small. In fact, there are at most $n$
cuts of the form $S = V_{\ell,i}(n)$ for $\ell\in\{1,\ldots, L(n)\}$,
and there are $n$ cuts of the form $S = \{u\}$ for $u\in V(n)$.  Hence
$\hLambdabuc(n)$ is described by at most $2n$ cuts.

The next theorem shows that $\hLambdabuc(n)$ is approximately (in the
scaling sense) equal to the balanced unicast capacity region
$\Lambdabuc(n)$ of the wireless network.
\begin{theorem}
    \label{thm:unicast}
    Under either fast or slow fading, for any $\alpha>2$, there exist
    \begin{align*}
        b_1(n) & \geq n^{-o(1)}, \\
        b_2(n) & = O(\log^6(n)),
    \end{align*}
    such that
    \begin{equation*}
        b_1(n) \hLambdabuc(n)
        \subset \Lambdabuc(n) 
        \subset b_2(n) \hLambdabuc(n),
    \end{equation*}
    with probability $1-o(1)$ as $n\to\infty$.
\end{theorem}
We point out that Theorem \ref{thm:unicast} holds only with probability
$1-o(1)$ for different reasons for the fast and slow fading cases. Under
fast fading, the theorem holds only for node placements that are
``regular enough''. The node placement itself is random, and we show
that the required regularity property is satisfied with high probability
as $n\to\infty$. Under slow fading, the theorem holds under the same
regularity requirements on the node placement, but now it also only
holds with high probability for the realization of the fading
$\{\theta_{u,v}\}_{u,v}$.

Theorem \ref{thm:unicast} provides a tight scaling characterization of
the entire balanced unicast capacity region $\Lambdabuc(n)$ of the
wireless network as depicted in Figure \ref{fig:general}. The
approximation is within a factor $n^{\pm o(1)}$. This factor can be
further sharpened as is discussed in detail in Section \ref{sec:second}.
\begin{figure}[!ht]
    \begin{center}
        \input{figs/capacity_region_approx.pstex_t}
    \end{center}

    \caption{
    The set $\hLambdabuc(n)$ approximates the balanced unicast capacity region
    $\Lambdabuc(n)$ of the wireless network in the sense that
    $b_1(n)\hLambdabuc(n)$ (with $b_1(n)\geq n^{-o(1)}$) provides
    an inner bound to $\Lambdabuc(n)$ and $b_2(n)\hLambdabuc(n)$
    (with $b_2(n) = O\big(\log^6(n)\big)$) provides an outer bound to
    $\Lambdabuc(n)$. The figure shows two dimensions (namely
    $\lambdauc_{1,2}$ and $\lambdauc_{2,1}$) of the $n^2$-dimensional set
    $\Lambdabuc(n)$.
    }

    \label{fig:general}
\end{figure}

We point out that for large values of path-loss exponent ($\alpha>5$) the
restriction to balanced traffic can be removed, yielding a tight scaling
characterization of the entire $n^2$-dimensional unicast capacity
region $\Lambdauc(n)$. See Section \ref{sec:large} for the details. For
$\alpha\in(2,5]$, bounds on achievable rates for traffic that is not
balanced are discussed in Section \ref{sec:nonbalanced}.

\subsection{Balanced Multicast Capacity Region}
\label{sec:multicast}

We now present an approximate characterization of the complete balanced
multicast capacity region $\Lambdabmc(n)$.  

Define
\begin{equation}
    \label{eq:approx_multicast}
    \begin{aligned}
        \hLambdamc(n)
        \defeq \Big\{\lambdamc\in\Rp^{n\times 2^n}: 
        & \sum_{u\in V_{\ell,i}(n)}
        \sum_{\substack{W\subset V(n):\\ W\setminus V_{\ell,i}(n)\neq\emptyset}}
        \lambdamc_{u,W}
        \leq (4^{-\ell}n)^{2-\min\{3,\alpha\}/2} \\
        & \qquad\qquad \forall \ell\in\{1,\ldots,L(n)\}, i\in\{1,\ldots, 4^\ell\}, \\
        & \sum_{\substack{W\subset V(n):\\ W\setminus\{u\}\neq\emptyset}}\lambdamc_{u,W}
        + \sum_{\tilde{u}\neq u}\sum_{\substack{W\subset V(n):\\ u\in W}}
        \lambdamc_{\tilde{u},W} \leq 1
        \ \forall u\in V(n)
        \Big\},
    \end{aligned}
\end{equation}
and set
\begin{equation*}
    \hLambdabmc(n) \defeq \hLambdamc(n)\cap\Bmc(n).
\end{equation*}
The definition of
$\hLambdabmc(n)$ is similar to the definition of $\hLambdabuc(n)$ in
\eqref{eq:approx_unicast}. $\hLambdabmc(n)$ is the collection of all
balanced multicast traffic matrices $\lambdamc$ such that for various cuts
$S\subset V(n)$ in the network, the total traffic demand (in either one
or both directions)
\begin{align*}
    & \sum_{u\in S} \sum_{\substack{W\subset V(n):\\ W\setminus S\neq\emptyset}}\lambdamc_{u,W}, \\
    & \sum_{u\in S} \sum_{\substack{W\subset V(n):\\ W\setminus S\neq\emptyset}}\lambdamc_{u,W}
    + \sum_{u\notin S} \sum_{\substack{W\subset V(n):\\ W\cap S\neq\emptyset}}\lambdamc_{u,W},
\end{align*}
across the cut $S$ is not too big. Note that, unlike in the definition
of $\hLambdabuc(n)$, we count $\lambda_{u,W}$ as crossing the cut $S$ 
if $u\in S$ and $W\setminus S\neq\emptyset$,
i.e., if there is at least one node $w$ in the multicast destination
group $W$ that lies outside $S$. The number of such cuts $S$ we need to
consider is at most $2n$, as in the unicast case.

The next theorem shows that $\hLambdabmc(n)$ is approximately (in the
scaling sense) equal to the balanced multicast capacity region
$\Lambdabmc(n)$ of the wireless network.
\begin{theorem}
    \label{thm:multicast}
    Under either fast or slow fading, for any $\alpha>2$, there exist
    \begin{align*}
        b_3(n) & \geq n^{-o(1)}, \\
        b_4(n) & = O(\log^6(n)),
    \end{align*}
    such that
    \begin{equation*}
        b_3(n) \hLambdabmc(n)
        \subset \Lambdabmc(n) 
        \subset b_4(n) \hLambdabmc(n),
    \end{equation*}
    with probability $1-o(1)$ as $n\to\infty$.
\end{theorem}
As with Theorem \ref{thm:unicast}, Theorem \ref{thm:multicast} holds
only with probability $1-o(1)$ for different reasons for the fast and
slow fading cases. 
Theorem \ref{thm:multicast} implies that the quantity $\hLambdabmc(n)$
determines the scaling of the balanced multicast capacity region
$\Lambdabmc(n)$.  The approximation is up to a factor $n^{\pm o(1)}$ as
in the unicast case, and can again be sharpened (see the discussion in
Section \ref{sec:second}). 
As in the unicast case, for $\alpha > 5$ the restriction of balanced
traffic can be dropped resulting in a scaling characterization of the
entire $n2^n$-dimensional multicast capacity region $\Lambdamc(n)$. The
details can be found in Section \ref{sec:large}. Similarly, we can
obtain bounds on achievable rates for traffic that is not balanced, as
is discussed in Section \ref{sec:nonbalanced}.

\subsection{Implications of Theorems \ref{thm:unicast} and \ref{thm:multicast}}
\label{sec:implications}

Theorems \ref{thm:unicast} and \ref{thm:multicast} can be applied in two
ways. First, the theorems can be used to analyze the asymptotic
achievability of a sequence of traffic matrices. Consider the unicast
case, and let $\{\lambdauc(n)\}_{n\geq 1}$ be a sequence of balanced unicast
traffic matrices with $\lambdauc(n)\in\Rp^{n\times n}$. Define
\begin{align*}
    \rho^{\star}_{\lambdauc}(n) & \defeq \sup\{\rho: \rho\lambdauc(n)\in\Lambdabuc(n)\}, \\
    \hat{\rho}^{\star}_{\lambdauc}(n) & \defeq \sup\{\hat{\rho}: \hat{\rho}\lambdauc(n)\in\hLambdabuc(n)\},
\end{align*}
i.e., $\rho^{\star}_{\lambdauc}(n)$ is the largest multiplier $\rho$ such that
the scaled traffic matrix $\rho\lambdauc(n)$  is contained in
$\Lambdabuc(n)$ (and similar for $\hat{\rho}^{\star}_{\lambdauc}(n)$ with
respect to $\hLambdabuc(n)$). Then Theorem \ref{thm:unicast} provides
asymptotic information about the achievability of
$\{\lambdauc(n)\}_{n\geq 1}$ in the sense that\footnote{We assume here
that the limits exist, otherwise the same statement holds for $\limsup$
and $\liminf$.}
\begin{align*}
    \lim_{n\to\infty}\frac{\log(\rho^{\star}_{\lambdauc}(n))}{\log(n)}
    = \lim_{n\to\infty}\frac{\log(\hat{\rho}^{\star}_{\lambdauc}(n))}{\log(n)}.
\end{align*}
Theorem \ref{thm:multicast} can be used similarly to analyze
sequences of balanced multicast traffic matrices. Several applications of this
approach are explored in Section \ref{sec:examples}. 

Second, Theorems \ref{thm:unicast} and \ref{thm:multicast} provide
information about the shape of the balanced unicast and multicast capacity
regions $\Lambdabuc(n)$ and $\Lambdabmc(n)$. Consider again the unicast
case. We now argue that even though the approximation $\hLambdabuc(n)$ of
$\Lambdabuc(n)$ is only up to $n^{\pm o(1)}$ scaling, its shape is largely
preserved. 

To illustrate this point, consider a rectangle
\begin{equation*}
    R(n)\defeq[0,r_1(n)]\times[0,r_2(n)],
\end{equation*}
and let
\begin{equation*}
    \widehat{R}(n)\defeq[0,\hat{r}_1(n)]\times[0,\hat{r}_2(n)],
\end{equation*}
where
\begin{equation*}
    \hat{r}_i \defeq b_i(n)r_i(n)
\end{equation*}
for some $b_i(n) = n^{\pm o(1)}$, be its approximation. The shape of $R(n)$
is then determined by the ratio between $r_1(n)$ and $r_2(n)$. 
For example, assume $r_1(n)=n^\beta r_2(n)$. Then
\begin{equation*}
    \frac{\hat{r}_1(n)}{\hat{r}_2(n)}
    = n^{\beta\pm o(1)}
    = n^{\pm o(1)}\frac{r_1(n)}{r_2(n)},
\end{equation*}
i.e.,
\begin{equation*}
    \lim_{n\to\infty}\frac{\log\big(r_1(n)/r_2(n)\big)}{\log(n)}
    = \beta
    = \lim_{n\to\infty}\frac{\log\big(\hat{r}_1(n)/\hat{r}_2(n)\big)}{\log(n)},
\end{equation*}
and hence the approximation $\widehat{R}(n)$ preserves the exponent of
the ratio of sidelengths of $R(n)$. In other words, if the two
sidelengths $r_1(n)$ and $r_2(n)$ differ on exponential scale (i.e., by
a factor $n^{\beta}$ for $\beta \neq 0$) then this shape information is
preserved by the approximation $\widehat{R}(n)$.

Let us now return to the balanced unicast capacity region
$\Lambdabuc(n)$ and its approximation $\hLambdabuc(n)$. We consider
several boundary points of $\Lambdabuc(n)$ and show that their behavior
varies at scale $n^\beta$ for various values of $\beta$.  From the
discussion in the previous paragraph, this implies that a significant
part of the shape of $\Lambdabuc(n)$ is preserved by its approximation
$\hLambdabuc(n)$.  First, let $\lambdauc\defeq\rho(n)\bm{1}$ for some
scalar $\rho(n)$ depending only on $n$, and where $\bm{1}$ is the
$n\times n$ matrix of all ones. If $\lambdauc\in\Lambdabuc(n)$ then the
largest achievable value of $\rho(n)$ is $\rho^{\star}(n) \leq
n^{-\min\{3,\alpha\}/2+o(1)}$ (by applying Theorem \ref{thm:unicast}).
Second, let $\lambdauc$ such that $\lambdauc_{u^{\star},w^{\star}}=
\lambdauc_{w^{\star},u^{\star}} = \rho(n)$ for only one
source-destination pair $(u^{\star},w^{\star})$ with $u^{\star}\neq
w^{\star}$ and $\lambdauc_{u,w}=0$ otherwise. Then $\rho^{\star}(n)$,
the largest achievable value of $\rho(n)$, satisfies
$\rho^{\star}(n)\geq n^{-o(1)}$. Hence the boundary points of
$\Lambdabuc(n)$ vary at least from $n^{-\min\{3,\alpha\}/2+o(1)}$ to
$n^{-o(1)}$, and this variation on exponential scale is preserved by
$\hLambdabuc(n)$.

Again, a similar analysis is possible also for the multicast capacity
region, showing that the approximate balanced multicast capacity region
$\hLambdabmc(n)$ preserves the shape of the balanced multicast capacity
region $\Lambdabmc(n)$ on exponential scale.

\subsection{Computational Aspects}
\label{sec:computation}

Since we are interested in large wireless networks, computational
aspects are of importance. In this section, we show that the approximate
characterizations $\hLambdabuc(n)$ and $\hLambdabmc(n)$ in Theorems
\ref{thm:unicast} and \ref{thm:multicast} provide a computationally
efficient approximate description of the balanced unicast and multicast
capacity regions $\Lambdabuc(n)$ and $\Lambdabmc(n)$, respectively.

Consider first the unicast case. Note that $\Lambdabuc(n)$ is a
$n^2$-dimensional set, and hence its shape could be rather complicated. In
particular, in the special cases where the capacity region is known, its
description is often in terms of cut-set bounds. Since there are $2^n$
possible subsets of $n$ nodes, there are $2^n$ possible cut-set bounds
to be considered. In other words, the description complexity of
$\Lambdabuc(n)$ is likely to be growing exponentially in $n$. On the
other hand, as was pointed out in Section \ref{sec:unicast}, the
description of $\hLambdabuc(n)$ is in terms of only $2n$ cuts. This
implies that $\hLambdabuc(n)$ can be computed efficiently (i.e., in
polynomial time in $n$). Hence even though the description complexity of
$\Lambdabuc(n)$ is likely to be of order $\Theta(2^n)$, the description
complexity of its approximation $\hLambdabuc(n)$ is only of order
$\Theta(n)$---an exponential reduction. In particular, this implies
that membership $\lambdauc\in\hLambdabuc(n)$ (and hence by Theorem
\ref{thm:unicast} also the approximate achievability of the balanced
unicast traffic matrix $\lambdauc$) can be computed in polynomial time
in the network size $n$. More precisely, evaluating each of the
$\Theta(n)$ cuts takes at most $\Theta(n^2)$ operations, yielding a
$\Theta(n^3)$-time algorithm for approximate testing of membership in
$\Lambdabuc(n)$.

Consider now the multicast case. $\Lambdabmc(n)$ is a $n
2^n$-dimensional set, i.e., the number of dimensions is exponentially
large in $n$. Nevertheless, its approximation $\hLambdabmc(n)$ can (as
in the unicast case) be computed by evaluating at most $2n$ cuts. This
yields a very compact approximate representation of the balanced
multicast capacity region $\Lambdabmc(n)$ (i.e., we represent a region
of exponential size in $n$ as an intersection of only linearly many
halfspaces---one halfspace corresponding to each cut).  Moreover, it
implies that membership $\lambdamc\in\hLambdabmc(n)$ can be computed
efficiently. More precisely, evaluating each of the $\Theta(n)$ cuts
takes at most $|\{(u,W) : \lambdamc_{u,W} > 0\}|$ operations. Thus
membership $\lambdamc\in\hLambdabmc(n)$ (and hence by Theorem
\ref{thm:multicast} also the approximate achievability of the balanced
multicast traffic matrix $\lambdamc$) can be tested in at most
$\Theta(n)$ times more operations than required to just read the problem
parameters. In other words, we have a linear time (in the length of the
input) algorithm for testing membership of a multicast traffic matrix
$\lambdamc$ in $\hLambdabmc(n)$, and hence for approximate testing of
membership in $\Lambdabmc(n)$. However, this algorithm is not
necessarily polynomial time in $n$, since reading just the input
$\lambdamc\in\Rp^{n\times 2^n}$ itself might take exponential time in
$n$.

\section{Example Scenarios}
\label{sec:examples}

We next illustrate the above results by determining achievable rates in
a few specific wireless network scenarios with non-uniform traffic
patterns. 

\begin{example}
    \emph{Multiple classes of source-destination pairs}
    \label{eg:separation}
    
    There are $K$ classes of source-destination pairs for some fixed
    $K$. Each source node in class $i$ generates traffic at the same
    rate $\rho_i(n)$ for a destination node that is chosen randomly
    within distance $\Theta(n^{\beta_i/2})$, for some fixed $\beta_i \in
    [0,1]$. Each node randomly picks the class it belongs to.  The
    resulting traffic matrix is balanced (with $\gamma(n)=n^{o(1)}$)
    with high probability, and applying Theorem~\ref{thm:unicast} shows
    that $\rho^{\star}_i(n)$, the largest achievable value of $\rho_i(n)$,
    satisfies
    \begin{equation*}
        \rho_i^{\star}(n) = n^{\beta_i (1-\ald/2) \pm o(1)},
    \end{equation*}
    with probability $1-o(1)$ for all $i$, and where 
    \begin{equation}
        \label{eq:alphabar}
        \ald \defeq \min\{3, \alpha\}.
    \end{equation}
    Hence, for a fixed number of classes $K$, source nodes in each class
    can obtain rates as a function of only the source-destination
    separation in that class. 
    
    Set $\tilde{n}_i\defeq n^{\beta_i}$, and note that $\tilde{n}_i$ is
    on the order of the expected number of nodes that are closer to a
    source than its destination. Then 
    \begin{equation*}
        \rho_i^{\star}(n) = n^{\pm o(1)}\tilde{n}_i^{1-\ald/2}.
    \end{equation*}
    Now $\tilde{n}_i^{1-\ald/2}$ is precisely the per-node rate that is
    achievable for an extended network with $\tilde{n}_i$ nodes under
    random source-destination pairing \cite{ozg}. In other words,
    the local traffic pattern here allows us to obtain a rate that is as
    good as the one achievable under random source-destination pairing
    for a much smaller network.
\end{example}

\begin{example}
    \emph{Traffic variation with source-destination separation}
    
    Assume each node is source for exactly one destination, chosen
    uniformly at random from among all the other nodes (as in the
    traditional setting). However, instead of all sources generating
    traffic at the same rate, source node $u$ generates traffic at a rate
    that is a function of its separation from destination $w$, i.e., the
    traffic matrix is given by $\lambdauc_{u,w} = \psi(r_{u,w})$ for
    some function $\psi$. In particular, let us consider 
    \begin{equation*}
        \psi(r) \defeq
        \rho(n)\times
        \begin{cases}
            r^\beta &  \text{if } r \geq 1, \\
            1 & \text{else},
        \end{cases}
    \end{equation*}
    for some fixed $\beta\in\R$ and some $\rho(n)$ depending only on
    $n$. The traditional setting corresponds to $\beta=0$,
    in which case all $n$ source-destination pairs communicate at
    uniform rate. 
    
    While such traffic is not balanced for small values of $\beta$, the
    results in Section \ref{sec:nonbalanced}, extending
    Theorem~\ref{thm:unicast} to traffic that is not balanced, can be used to
    establish the scaling of $\rho^{\star}(n)$, the largest achievable
    value of $\rho(n)$, as
    \begin{equation*}
        \rho^{\star}(n) = 
        \begin{cases}
            n^{1-(\ald + \beta)/2 \pm o(1)} & \text{if $\beta \geq 2-\ald$}, \\
            n^{\pm o(1)} & \text{else},
        \end{cases}
    \end{equation*}
    with probability $1-o(1)$. For $\beta = 0$, and noting that
    $2-\ald\leq 0$, this recovers the results from \cite{ozg} for
    random source-destination pairing with uniform rate. 
\end{example}

\begin{example}
    \emph{Sources with multiple destinations}
    
    All the example scenarios so far are concerned with traffic in which
    each node is source exactly once. Here we consider more general
    traffic patterns. There are $K$ classes of source nodes, for some
    fixed $K$. Each source node in class $i$ has $\Theta(n^{\beta_i})$
    destination nodes for some fixed $\beta_i \in [0,1]$ and generates
    independent traffic at the same rate $\rho_i(n)$ for each of them
    (i.e., we still consider unicast traffic). Each of these destination
    nodes is chosen uniformly at random among the $n-1$ other nodes.
    Every node randomly picks the class it belongs to.  Noting that the
    resulting traffic matrix is balanced with high
    probability, Theorem~\ref{thm:unicast} provides the following
    scaling of the rates achievable by different classes:
    \begin{equation*}
        \rho_i^{\star}(n) = n^{1-\beta_i - \ald/2 \pm o(1)},
    \end{equation*}
    with probability $1-o(1)$ for all $i$. In other words, for each
    source node time sharing between all $K$ classes and then (within
    each class) between all its $\Theta(n^{\beta_i})$ destination nodes
    is order-optimal in this scenario. However, different sources are
    operating simultaneously.
\end{example}

\begin{example}
    \label{eg:multicast2}
    \emph{Broadcast}

    Consider a scenario with every node $u$ in the network broadcasting
    an independent message to all other nodes at rate $\rho(n)
    \lambda_u$.  In other words, we have a multicast traffic matrix of
    the form
    \begin{equation*}
        \lambdamc_{u,W} =
        \begin{cases}
            \rho(n)\lambda_u & \text{if $W=V(n)$}, \\
            0 & \text{else},
        \end{cases}
    \end{equation*}
    for some $\rho(n)>0$. Applying the generalization in Section
    \ref{sec:nonbalanced} of Theorem \ref{thm:multicast} yields that
    $\rho^{\star}(n)$, the largest achievable $\rho(n)$, satisfies
    \begin{equation*}
        \rho^{\star}(n) = n^{\pm o(1)}\frac{1}{\sum_{u\in V(n)}\lambda_u}
    \end{equation*}
    as $n\to\infty$.
\end{example}

\section{Communication Schemes}
\label{sec:schemes}

In this section, we provide a high-level description of the
communication schemes used to prove achievability (i.e., the inner
bound) in Theorems \ref{thm:unicast} and \ref{thm:multicast}. In Section
\ref{sec:schemes_unicast}, we present a communication scheme for general
unicast traffic, in Section \ref{sec:schemes_multicast} we show how this
scheme can be adapted for general multicast traffic. Both schemes use as
a building block a communication scheme introduced in prior work for a
particular class of traffic, called \emph{uniform permutation traffic}.
In such uniform permutation traffic, each node in the network is source
and destination exactly once, and all these source-destination pairs
communicate at equal rate. For $\alpha\in (2,3]$, the order-optimal
scheme for such uniform permutation traffic (called \emph{hierarchical
relaying scheme} in the following) enables global cooperation in the
network. For $\alpha > 3$, the order-optimal scheme is multi-hop
routing.  We recall these two schemes for uniform permutation traffic in
Section \ref{sec:schemes_hr}.

\subsection{Communication Scheme for Unicast Traffic}
\label{sec:schemes_unicast}

In this section, we present a scheme to transmit general unicast
traffic. This scheme has a tree structure that makes it convenient to
work with. This tree structure is crucial in proving the compact
approximation of the balanced unicast capacity region $\Lambdabuc(n)$ in
Theorem \ref{thm:unicast}.

The communication scheme consists of three layers: A top or routing
layer, a middle or cooperation layer, and a bottom or physical layer.
The routing layer of this scheme treats the wireless network as a tree
graph $G$ and routes messages between sources and their
destinations---dealing with heterogeneous traffic demands. The
cooperation layer of this scheme provides this tree abstraction $G$ to
the top layer by appropriately distributing and concentrating traffic
over the wireless network---choosing the level of cooperation in the
network. The physical layer implements this distribution and
concentration of messages in the wireless network---dealing with
interference and noise. 

Seen from the routing layer, the network consists of a noiseless
capacitated graph $G$. This graph is a tree, whose leaf nodes are
the nodes $V(n)$ in the wireless network. The internal nodes of $G$
represent larger clusters of nodes (i.e., subsets of $V(n)$) in the
wireless network. More precisely, each internal node in $G$ represents a
set $V_{\ell,i}(n)$ for $\ell\in\{1,\ldots,L(n)\}$ and $i\in\{1,\ldots,
4^\ell\}$. Consider two sets $V_{\ell,i}(n), V_{\ell+1,j}(n)$ and let
$\nu, \mu$ be the corresponding internal nodes in $G$. Then $\nu$ and
$\mu$ are connected by an edge in $G$ if $V_{\ell+1,j}(n)\subset
V_{\ell,i}(n)$. Similarly, for $V_{L(n),i}(n)$ and corresponding
internal node $\nu$ in $G$, a leaf node $u$ in $G$ is connected by an
edge to $\nu$ if $u\in V_{L(n),i}(n)$ (recall that the leaf nodes of $G$
are the nodes $V(n)$ in the wireless network). This construction is
shown in Figure \ref{fig:grid_graph}.
\begin{figure}[!ht]
    \begin{center}
        \scalebox{0.888}{\input{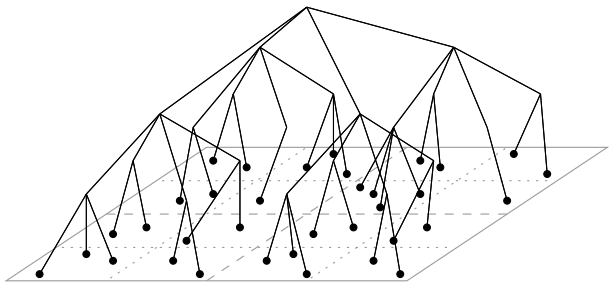}}
    \end{center}

    \caption{Construction of the tree graph $G$. We consider the same
    nodes as in Figure \ref{fig:grid} with $L(n) = 2$. The leaves of $G$
    are the nodes $V(n)$ of the wireless network. They are always at
    level $\ell=L(n)+1$ (i.e., $3$ in this example). At level $0\leq\ell\leq
    L(n)$ in $G$, there are $4^\ell$ nodes. The tree structure is the
    one induced by the grid decomposition $\{V_{\ell,i}(n)\}_{\ell,i}$ as
    shown in Figure \ref{fig:grid}. Level $0$ contains the root node of $G$.
    }

    \label{fig:grid_graph}
\end{figure}
In the routing layer, messages are sent from each source to its 
destination by routing them over $G$. To send information along an
edge of $G$, the routing layer calls upon the cooperation layer. 

The cooperation layer implements the tree abstraction $G$. This is done
by ensuring that whenever a message is located at a node in $G$, it is
evenly distributed over the corresponding cluster in the wireless
network, i.e., every node in the cluster has access to a distinct part
of equal length of the message. To send information from a child node to
its parent in $G$ (i.e., towards the root node of $G$), the message at
the cluster in $V(n)$ represented by the child node is distributed
evenly among all nodes in the bigger cluster in $V(n)$ represented by
the parent node.  More precisely, let $\nu$ be a child node of $\mu$ in
$G$, and let $V_{\ell+1,i}(n), V_{\ell,j}(n)$ be the corresponding subsets
of $V(n)$. Consider the cooperation layer being called by the routing
layer to send a message from $\nu$ to its parent $\mu$ over $G$. In the
wireless network, we assume each node in $V_{\ell+1,i}(n)$ has access to
a distinct $1/\card{V_{\ell+1,i}(n)}$ fraction of the message to be
sent. Each node in $V_{\ell+1,i}(n)$ splits its message part into four
distinct parts of equal length. It keeps one part for itself and sends
the other three parts to three nodes in $V_{\ell,j}(n)\setminus
V_{\ell+1,i}(n)$. After each node in $V_{\ell+1,i}(n)$ has sent its
message parts, each node in $V_{\ell,j}(n)$ now as access to a distinct
$1/\card{V_{\ell,j}(n)}$ fraction of the message. To send information
from a parent node to a child node in $G$ (i.e., away from the root node
of $G$), the message at the cluster in $V(n)$ represented by the parent
node is concentrated on the cluster in $V(n)$ represented by the child
node. More precisely, consider the same nodes $\nu$ and $\mu$ in $G$
corresponding to $V_{\ell+1,i}(n)$ and $V_{\ell,j}(n)$ in $V(n)$.
Consider the cooperation layer being called by the routing layer to send
a message from $\mu$ to its child $\nu$. In the wireless network, we
assume each node in $V_{\ell,j}(n)$ has access to a distinct
$1/\card{V_{\ell,j}(n)}$ fraction of the message to be sent.  Each node
in $V_{\ell,j}(n)$ sends its message part to another node in
$V_{\ell+1,i}(n)$. After each node in $V_{\ell,j}(n)$ has sent its
message part, each node in $V_{\ell+1,i}(n)$ now as access to a distinct
$1/\card{V_{\ell+1,i}(n)}$ fraction of the message.  To implement this
distribution and concentration of messages, the cooperation layer calls
upon the physical layer.

The physical layer performs the distribution and concentration of
messages. Note that the traffic induced by the cooperation layer in the
physical layer is very regular, and closely resembles a uniform
permutation traffic (in which each node in the wireless network is
source and destination once and all these source-destination pairs want
to communicate at equal rate). Hence we can use either cooperative
communication (for $\alpha\in(2,3]$) or multi-hop communication (for
$\alpha>3$) for the transmission of this traffic. See Section
\ref{sec:schemes_hr} for a detailed description of these two schemes. It
is this operation in the physical layer that determines the edge
capacities of the graph $G$ as seen from the routing layer.

The operation of this three-layer architecture is illustrated in the
following example.
\begin{example}
    Consider a single source-destination pair $(u,w)$. The corresponding
    operation of the three-layer architecture is depicted in Figure
    \ref{fig:layers}. 
    \begin{figure}[!ht]
        \begin{center}
            \scalebox{0.888}{\input{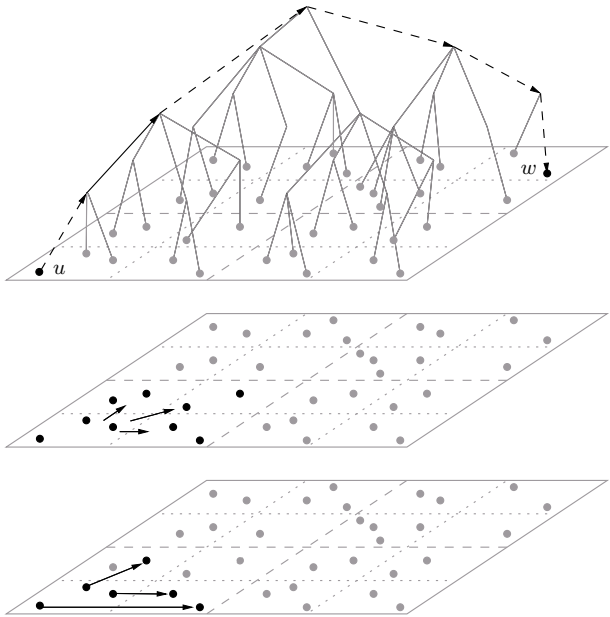}}
        \end{center}

        \caption{Example operation of the three-layer architecture under
        unicast traffic. The three layers depicted are (from top to bottom
        in the figure) the routing layer, the cooperation layer, and the
        physical layer.}

        \label{fig:layers}
    \end{figure}

    In the routing layer, the message is routed over the tree graph $G$
    between $u$ and $w$ (indicated in black in the figure). The middle
    plane in the figure shows the induced behavior from using the second
    edge along this path (indicated in solid black in the figure) in the
    cooperation layer.  The bottom plane in the figure shows (part of)
    the corresponding actions induced in the physical layer. Let us now
    consider the specific operations of the three layers for the single
    message between $u$ and $w$. Since $G$ is a tree, there is a unique
    path between $u$ and $w$, and the routing layer sends the message
    over the edges along this path. Consider now the first such edge.
    Using this edge in the routing layer induces the following actions
    in the cooperation layer. The node $u$, having access to the entire
    message, splits that message into $3$ distinct parts of equal
    length. It keeps one part, and sends the other two parts to the two
    other nodes in $V_{2,1}(n)$ (i.e., lower left square at level
    $\ell=2$ in the hierarchy). In other words, after the message has
    traversed the edge between $u$ and its parent node in the routing
    layer, all nodes in $V_{2,1}(n)$ in the cooperation layer have
    access to a distinct $1/3$ fraction of the original message. The
    edges in the routing layer leading up the tree (i.e., towards the
    root node) are implemented in the cooperation layer in a similar
    fashion by further distributing the message over the wireless
    network. By the time the message reaches the root node of $G$ in the
    routing layer, the cooperation layer has distributed the message
    over the entire network and every node in $V(n)$ has access to a
    distinct $1/n$ fraction of the original message. Communication down
    the tree in the routing layer is implemented in the cooperation
    layer by concentrating messages over smaller regions in the wireless
    network.  To physically perform this distribution and concentration
    of messages, the cooperation layer calls upon the physical layer,
    which uses either hierarchical relaying or multi-hop communication.
\end{example}

\subsection{Communication Scheme for Multicast Traffic}
\label{sec:schemes_multicast}

Here we show that the same communication scheme presented in the last
section for general unicast traffic can also be used to transmit general
multicast traffic. Again it is the tree structure of the scheme that is
critically exploited in the proof of Theorem \ref{thm:multicast}
providing an approximation for the balanced multicast capacity region
$\Lambdabmc(n)$.

We will use the same three-layer architecture as for unicast traffic
presented in Section \ref{sec:schemes_unicast}. To accommodate multicast
traffic, we only modify the operation of the top or routing layer; the
lower layers operate as before. 

We now outline how the routing layer needs to be adapted for the
multicast case.  Consider a multicast message that needs to be
transmitted from a source node $u\in V(n)$ to its set of intended
destinations $W\subset V(n)$. In the routing layer, we want to route
this message from $u$ to $W$ over $G$. Since $G$ is a tree,
the routing part is simple. In fact, between $u$ and every $w\in W$
there exists a unique path in $G$. Consider the union of all those
paths. It is easy to see that this union is a subtree of $G$. Indeed, it
is the smallest subtree of $G$ that covers $\{u\}\cup W$. Traffic is
optimally routed over $G$ from $u$ to $W$ by sending it along the edges
of this subtree. 

The next example illustrates the operation of the routing layer under
multicast traffic.
\begin{example}
    Consider one source node $u$ and the corresponding multicast group
    $W\defeq\{w_1,w_2,w_3\}$ as shown in Figure \ref{fig:grid_graph_example}.
    \begin{figure}[!ht]
        \begin{center}
            \scalebox{0.888}{\input{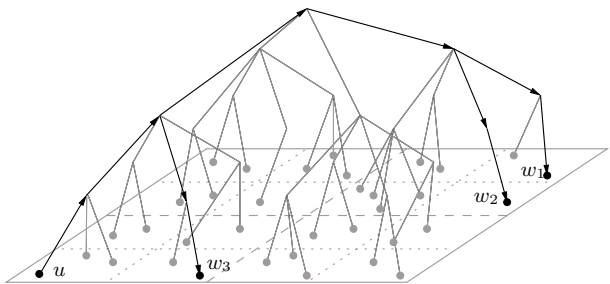}}
        \end{center}

        \caption{Example operation of the routing layer in the three-layer
        architecture under multicast traffic.}

        \label{fig:grid_graph_example}
    \end{figure}

    In the routing layer, we find the smallest subgraph $G(\{u\}\cup W)$
    covering $\{u\}\cup W$ (indicated by black lines in Figure
    \ref{fig:grid_graph_example}).  Messages are sent from the source to
    its destinations by routing them along this subgraph.  In other
    words, $G(\{u\}\cup W)$ is the multicast tree along which the
    message is sent from $u$ to $W$. The cooperation layer and physical
    layer operate in the same way as for unicast traffic (see Figure
    \ref{fig:layers} for an example).  
\end{example}

\subsection{Communication Schemes for Uniform Permutation Traffic}
\label{sec:schemes_hr}

Here we recall communication schemes for uniform permutation traffic on
$A(n)$, i.e., each node is source and destination exactly once and all
these $n$ pairs communicate at uniform rate. As pointed out in Sections
\ref{sec:schemes_unicast} and \ref{sec:schemes_multicast}, these
communication schemes are used as building blocks in the communication
architecture for general unicast and multicast traffic.

The structure of the optimal communication scheme depends drastically on
the path-loss exponent $\alpha$. For $\alpha\in(2,3]$ (small path-loss
exponent), cooperative communication on a global scale is necessary to
achieve optimal performance. For $\alpha >3$ (large path-loss exponent),
local communication between neighboring nodes is sufficient, and traffic
is routed in a multi-hop fashion from the source to the destination. We
will refer to the order-optimal scheme for $\alpha\in(2,3]$ as
\emph{hierarchical relaying scheme}, and to the order optimal scheme for
$\alpha > 3$ as \emph{multi-hop scheme}. For a uniform permutation traffic on
$V(n)$, hierarchical relaying achieves a per-node rate of
$n^{1-\alpha/2-o(1)}$; multi-hop communication achieves a per-node rate
of $n^{-1/2-o(1)}$. By choosing the appropriate scheme (hierarchical
relaying for $\alpha\in(2,3]$, multi-hop for $\alpha>3$), we can thus
achieve a per-node rate of $n^{1-\min\{3,\alpha\}/2-o(1)}$.  We provide
a short description of the hierarchical relaying scheme in the
following. The details can be found in \cite{nie}.

Consider $n$ nodes placed independently and uniformly at random on $A(n)$.
Divide $A(n)$ into
\begin{equation*}
    n^{\tfrac{1}{1+\log^{1/3}(n)}}
\end{equation*}
squarelets of equal size. Call a squarelet \emph{dense}, if it contains
a number of nodes proportional to its area. For each source-destination
pair, choose such a dense squarelet as a \emph{relay}, over which it
will transmit information (see Figure \ref{fig:relay}).

\begin{figure}[!ht]
    \begin{center}
        \input{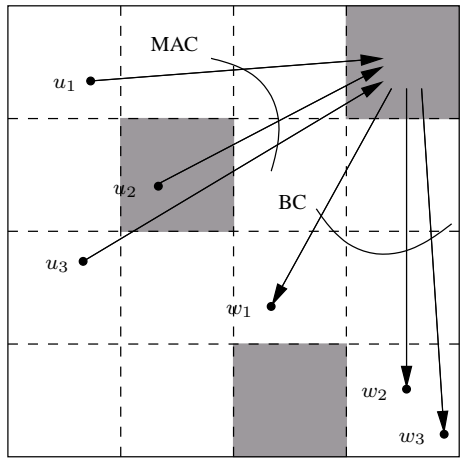}
    \end{center}

    \caption{Sketch of one level  of the hierarchical relaying scheme. Here
    $\{(u_i,w_i)\}_{i=1}^3$ are three source-destination pairs.  Groups
    of source-destination pairs relay their traffic over dense
    squarelets (shaded), which contain a number of nodes proportional to their
    area.  We time share between the different relay
    squarelets. Within each relay squarelet the scheme is used
    recursively to enable joint decoding and encoding at the relay.
    }

    \label{fig:relay}
\end{figure}

Consider now one such relay squarelet and the nodes that are
transmitting information over it.  If we assume for the moment that the
nodes within the relay squarelets could cooperate, then between the
source nodes and the relay squarelet we would have a multiple access
channel (MAC), where each source node has one transmit antenna, and
the relay squarelet (acting as one node) has many receive antennas.
Between the relay squarelet and the destination nodes, we would have a
broadcast channel (BC), where each destination node has one receive
antenna, and the relay squarelet (acting again as one node) has many
transmit antennas.  The cooperation gain from using this kind of scheme
arises from the use of multiple antennas for this MAC and BC.

To actually enable this kind of cooperation at the relay squarelet,
local communication within the relay squarelets is necessary. It can be
shown that this local communication problem is actually the same as the
original problem, but at a smaller scale. Indeed, we are now considering
a square of size
\begin{equation*}
    n^{1-\tfrac{1}{1+\log^{1/3}(n)}}
\end{equation*}
with equal number of nodes (at least order wise). Hence we can use the
same scheme recursively to solve this subproblem. We terminate the
recursion after $\log^{1/3}(n)$ iterations, at which point we use simple
time-division multiple access (TDMA) to bootstrap the scheme.

Observe that at the final level of the scheme, we have divided $A(n)$
into
\begin{equation*}
    \Big(n^{\tfrac{1}{1+\log^{1/3}(n)}}\Big)^{ \log^{1/3}(n)}
    = n^{\tfrac{1}{1+\log^{-1/3}(n)}}
\end{equation*}
squarelets. A sufficient condition for the scheme to succeed is that all
these squarelets are dense (i.e., contain a number of nodes proportional
to their area). However much weaker conditions are sufficient as well,
see \cite{nie}.

For any permutation traffic, the per-node rate achievable with this
scheme is at least $n^{1-\alpha/2-o(1)}$
for any $\alpha>2$ and under fast fading. Under slow fading the same
per-node rate is achievable for all permutation traffic with probability
at least
\begin{equation*}
    1-\exp\Big(-2^{\Omega(\log^{2/3}(n))}\Big).
\end{equation*}
Moreover, when $\alpha\in(2,3]$ and for uniform permutation traffic with
a constant fraction of source-destination pairs at distance
$\Theta(\sqrt{n})$ (as is the case with high probability if the
permutation traffic is chosen at random), this is asymptotically the
best uniformly achievable per-node rate.

\section{Auxiliary Lemmas}
\label{sec:aux}

In this section, we provide auxiliary results, which will be used
several times in the following. These results are grouped into three
parts. In Section \ref{sec:aux_regular}, we describe regularity
properties exhibited with high probability by the random node placement.
In Section \ref{sec:aux_converse}, we provide auxiliary upper bounds on
the performance of any scheme in terms of cut-set bounds. Finally, in
Section \ref{sec:aux_achievability}, we describe auxiliary results on
the performance of hierarchical relaying and multi-hop communication as
described in Section \ref{sec:schemes_hr}.

\subsection{Regularity Lemmas}
\label{sec:aux_regular}

Here we prove several regularity properties that are satisfied with high
probability by a random node placement. Formally, define $\mc{V}(n)$ to
be the collection of all node placements $V(n)$ that satisfy the
following conditions:
\begin{align*}
    r_{u,v} & > n^{-1}
    & \text{for all $u,v\in V(n)$,} \\
    \big\lvert V_{\ell,i}(n) \big\rvert & \leq \log(n) 
    & \text{for $\ell= \frac{1}{2}\log(n)$,} \\
    \big\lvert V_{\ell,i}(n) \big\rvert & \geq 1 
    & \text{for $\ell=\frac{1}{2}\log\Big(\frac{n}{2\log(n)}\Big)$,} \\
    \big\lvert V_{\ell,i}(n) \big\rvert & \in [4^{-\ell-1}n,4^{-\ell+1}n] 
    & \text{for all $\ell \in\Big\{1,\ldots,L'(n)\Big\}$,}
\end{align*}
where 
\begin{equation*}
        L'(n)\defeq \frac{1}{2}\log(n)\big(1-\tfrac{1}{2}\log^{-5/6}(n)\big),
\end{equation*}
and in each case $i\in\{1,\ldots,4^\ell\}$.
The first condition is that the minimum distance between node pairs is
not too small. The second condition is that all squares of area $1$
contain at most $\log(n)$ nodes. The third condition is that all squares
of area $2\log(n)$ contain at least one node. The fourth condition is
that all squares up to level
$\frac{1}{2}\log(n)\big(1-\tfrac{1}{2}\log^{-5/6}(n)\big)$ contain a number of nodes
proportional to their area. Note that, since
\begin{align*}
    L(n) & = \frac{1}{2}\log(n)\big(1-\log^{-1/2}(n)\big) \\
    & = \frac{1}{2}\log(n)\big(1-\tfrac{1}{2}\log^{-5/6}(n)\big),
\end{align*}
this holds in particular for nodes up to level $L(n)$.  The goal of this
section is to prove that
\begin{equation*}
    \Pp(V(n)\in\mc{V}(n)) = 1-o(1),
\end{equation*}
as $n\to\infty$. 

The first lemma shows that the minimum distance in a
random node placement is at least $n^{-1}$ with high probability.
\begin{lemma}
    \label{thm:dmin}
    \begin{equation*}
        \Pp\Big(\min_{u\in V(n), v\in V(n)\setminus\{u\} } r_{u,v} > n^{-1}\Big) = 1-o(1),
    \end{equation*}
    as $n\to\infty$.
\end{lemma}
\begin{proof}
    For $u,v\in V$, let
    \begin{equation*}
        B_{u,v} \defeq \{r_{u,v} \leq r\}
    \end{equation*}
    for some $r$ (depending only on $n$). Fix a node $u\in V$, then for
    $v\neq u$
    \begin{equation*}
        \Pp(B_{u,v}| u )
        \leq \frac{r^2\pi}{n}
    \end{equation*}
    (the inequality being due to boundary effects). Moreover, the events
    $\{B_{u,v}\}_{v\in V\setminus\{u\}}$ are independent conditioned on
    $u$, and thus
    \begin{align*}
        \Pp\Big(\cap_{v\in V\setminus\{u\}}B_{u,v}^c\big\vert u\Big)
        & = \prod_{v\in V\setminus\{u\}}\Pp(B_{u,v}^c|u) \\
        & \geq \Big(1-\frac{r^2\pi}{n}\Big)^n.
    \end{align*}
    From this,
    \begin{align*}
        \Pp\Big(\min_{u\in V,v\in V\setminus\{u\}} r_{u,v} \leq r \Big)  
        & = \Pp\Big(\cup_{u\in V,v\in V\setminus\{u\}} B_{u,v} \Big) \\
        & \leq \sum_{u\in V} \Pp\Big(\cup_{v\in V\setminus\{u\}} B_{u,v} \Big) \\
        & = \sum_{u\in V} \bigg(1-\Pp\Big(\cap_{v\in V\setminus\{u\}} B_{u,v}^c\Big) \bigg) \\
        & = \sum_{u\in V} \bigg(1-\E\Big(\Pp\Big(\cap_{v\in V\setminus\{u\}} B_{u,v}^c\big\vert u \Big) \Big) \bigg) \\
        & \leq \sum_{u\in V} \Big(1-\Big(1-\frac{r^2\pi}{n}\Big)^n \Big) \\
        & = n\Big(1-\Big(1-\frac{r^2\pi}{n}\Big)^n \Big).
    \end{align*}
    Assuming $r < \sqrt{n/\pi}$, we have
    \begin{equation*}
        n\Big(1-\Big(1-\frac{r^2\pi}{n}\Big)^n \Big)
        \leq nr^2\pi,
    \end{equation*}
    and hence
    \begin{equation*}
        \Pp\Big(\min_{u\in V,v\in V\setminus\{u\}} r_{u,v} \leq r \Big) 
        \leq nr^2\pi,
    \end{equation*}
    which converges to zero for $r= n^{-1}$.
\end{proof}

The next lemma asserts that if $\widetilde{L}(n)$ is not too large
then all squares $\{V_{\ell,i}(n)\}_{\ell,i}$ for
$\ell\in\{1,\ldots,\widetilde{L}(n)\}$ and $i\in\{1,\ldots, 4^{\ell}\}$
in the grid decomposition of $V(n)$ contain a number of nodes that is
proportional to their area.

\begin{lemma}
    \label{thm:regular}
    If $\widetilde{L}(n)$ satisfies
    \begin{equation*}
        \lim_{n\to\infty}\frac{\widetilde{L}(n)}{4^{-\widetilde{L}(n)}n}=0
    \end{equation*}
    then
    \begin{equation*}
        \Pp\bigg(
        \bigcap_{\ell=1}^{\widetilde{L}(n)} \bigcap_{i=1}^{4^{\ell}}
        \big\{\card{V_{\ell,i}(n)}\in [4^{-\ell-1}n,4^{-\ell+1}n]\big\}
        \bigg)
        = 1-o(1)
    \end{equation*}
    as $n\to\infty$. In particular, this holds for
    \begin{equation*}
        \widetilde{L}(n)=\frac{1}{2}\log(n)\big(1-\tfrac{1}{2}\log^{-5/6}(n)\big),
    \end{equation*}
    and for $\widetilde{L}(n)=L(n)$.
\end{lemma}
\begin{proof}
    Let $B_u$ be the event that node $u$ lies in $A_{\ell,i}$ for fixed
    $\ell,i$. Note that
    \begin{equation*}
        \sum_{u\in V} \ind_{B_u} = \card{V_{\ell,i}}
    \end{equation*}
    by definition, and that 
    \begin{equation*}
        \Pp(B_u) = 4^{-\ell}.
    \end{equation*}
    Hence, using the Chernoff bound,
    \begin{align*}
        \Pp\big( \card{V_{\ell,i}} \not\in [4^{-\ell-1}n, 4^{-\ell+1}n] \big) 
        & = \Pp\Big( \sum_{u\in V} \ind_{B_u} \not\in [4^{-\ell-1}n,4^{-\ell+1}n] \Big) \\
        & \leq \exp(-K 4^{-\ell}n),
    \end{align*}
    for some positive constant $K$, and we obtain, for $\ell=\widetilde{L}(n)$,
    \begin{align}
        \label{eq:chernoff}
        \Pp\bigg(
        \bigcap_{i=1}^{4^{\widetilde{L}(n)}} &
        \big\{\lvert V_{\widetilde{L}(n),i}\rvert\in [4^{-\widetilde{L}(n)-1}n,4^{-\widetilde{L}(n)+1}n]\big\}
        \bigg) \nonumber\\
        & \geq 1 - \sum_{i=1}^{4^{\widetilde{L}(n)}}
        \Pp\big(\vert V_{\widetilde{L}(n),i}\rvert 
        \not\in [4^{-\widetilde{L}(n)-1}n,4^{-\widetilde{L}(n)+1}n]\big) \nonumber \\
        & \geq 1-4^{\widetilde{L}(n)}\exp(-K4^{-\widetilde{L}(n)}n) \nonumber \\
        & \geq 1-\exp(\widetilde{K}\widetilde{L}(n)-K4^{-\widetilde{L}(n)}n),
    \end{align}
    for some positive constant $\widetilde{K}$. By assumption
    \begin{equation*}
        \lim_{n\to\infty}\frac{\widetilde{L}(n)}{4^{-\widetilde{L}(n)}n}=0,
    \end{equation*}
    and hence
    \begin{equation*}
        \Pp\bigg(
        \bigcap_{i=1}^{4^{\widetilde{L}(n)}}
        \big\{\lvert V_{\widetilde{L}(n),i}\rvert\in [4^{-\widetilde{L}(n)-1}n,4^{-\widetilde{L}(n)+1}n]\big\}
        \bigg)
        \geq 1-o(1),
    \end{equation*}
    as $n\to\infty$.  Since the $\{A_{\ell,i}\}_{\ell,i}$ are nested as
    a function of $\ell$, we have
    \begin{equation*}
        \bigcap_{\ell=1}^{\widetilde{L}(n)} \bigcap_{i=1}^{4^{\ell}}
        \big\{\card{V_{\ell,i}}\in [4^{-\ell-1}n,4^{-\ell+1}n]\big\} 
        = \bigcap_{i=1}^{4^{\widetilde{L}(n)}}
        \big\{\lvert V_{\widetilde{L}(n),i} \rvert \in [4^{-\widetilde{L}(n)-1}n,4^{-\widetilde{L}(n)+1}n]\big\},
    \end{equation*}
    which, combined with \eqref{eq:chernoff}, proves the first part of
    the lemma.

    For the second part, note that for 
    \begin{equation*}
        \widetilde{L}(n) = \frac{1}{2}\log(n)\big(1-\tfrac{1}{2}\log^{-5/6}(n)\big),
    \end{equation*}
    we have
    \begin{align*}
        \frac{\widetilde{L}(n)}{4^{-\widetilde{L}(n)}n}
        & = \frac{\frac{1}{2}\log(n)\big(1-\tfrac{1}{2}\log^{-5/6}(n)\big)}
        {2^{\tfrac{1}{2}\log^{1/6}(n)}} \\
        & \leq \frac{\log(n)}{2^{\tfrac{1}{2}\log^{1/6}(n)}} \\
        & = 2^{\log\log(n)-\tfrac{1}{2}\log^{1/6}(n)} \to 0,
    \end{align*}
    and hence the lemma is valid in this case. The same holds for
    $\widetilde{L}(n)=L(n)$ since
    \begin{equation*}
        L(n) \leq \frac{1}{2}\log(n)\big(1-\tfrac{1}{2}\log^{-5/6}(n)\big).
    \end{equation*}
\end{proof}

We are now ready to prove that a random node placement $V(n)$ is in
$\mc{V}(n)$ with high probability as $n\to\infty$ (i.e., is fairly
``regular'' with high probability).

\begin{lemma}
    \label{thm:mcv}
    \begin{equation*}
        \Pp(V(n)\in\mc{V}(n)) = 1-o(1),
    \end{equation*}
    as $n\to\infty$.
\end{lemma}
\begin{proof}
    The first condition,
    \begin{align*}
        r_{u,v} & > n^{-1}
        & \text{for all $u,v\in V$,}
    \end{align*}
    holds with probability $1-o(1)$ by Lemma \ref{thm:dmin}. The second
    and third conditions,
    \begin{align*}
        \big\lvert V_{\ell,i} \big\rvert & \leq \log(n) 
        & \text{for $\ell= \frac{1}{2}\log(n)$,} \\
        \big\lvert V_{\ell,i} \big\rvert & \geq 1 
        & \text{for $\ell=\frac{1}{2}\log\Big(\frac{n}{2\log(n)}\Big)$,}
    \end{align*}
    are shown in \cite[Lemma 5.1]{ozg} to hold with probability
    $1-o(1)$. The fourth condition,
    \begin{align*}
        \big\lvert V_{\ell,i} \big\rvert & \in [4^{-\ell-1}n,4^{-\ell+1}n] 
        & \text{for all $\ell \in\Big\{1,\ldots, L'(n)\Big\}$,}
    \end{align*}
    holds with probability $1-o(1)$ by Lemma \ref{thm:regular}.
    Together, this proves the result.
\end{proof}

\subsection{Converse Lemmas}
\label{sec:aux_converse}

Here we prove several auxiliary converse results. The first lemma
bounds the maximal achievable sum rate for every individual node (i.e.,
the total traffic for which a fixed node is either source or
destination). 

\begin{lemma}
    \label{thm:cutset2}
    Under either fast or slow fading, for any $\alpha>2$, there exists
    $b(n)=O(\log(n))$ such that for all $V(n)\in\mc{V}(n)$,
    $\lambdauc\in\Lambdauc(n)$, $u\in V(n)$,
    \begin{align}
        \label{eq:thm3}
        \sum_{w\in V(n)\setminus\{u\}}\lambdauc_{u,w} \leq b(n), \\
        \label{eq:thm4}
        \sum_{w\in V(n)\setminus\{u\}}\lambdauc_{w,u} \leq b(n).
    \end{align}
\end{lemma}
\begin{proof}
    The argument follows the one in \cite[Theorem 3.1]{ozg}.  Denote by
    $C(S_1,S_2)$ the multiple-input multiple-output (MIMO) capacity
    between nodes in $S_1$ and nodes in $S_2$, for $S_1, S_2\subset V$.
    Consider first \eqref{eq:thm3}. By the cut-set bound \cite[Theorem
    14.10.1]{cov}, 
    \begin{equation*}
        \sum_{w\neq u}\lambdauc_{u,w} \leq C(\{u\},\{u\}^c).
    \end{equation*}
    $C(\{u\},\{u\}^c)$ is the capacity between $u$ and
    the nodes in $\{u\}^c$, i.e., 
    \begin{align*}
        C(\{u\},\{u\}^c) 
        & = \log\Big(1+ { \textstyle \sum_{v\neq u} } \abs{h_{u,v}}^2\Big) \\
        & \leq \log(1+(n-1)n^\alpha) \\
        & \leq K\log(n),
    \end{align*}
    with 
    \begin{equation*}
        K \defeq 2+\alpha, 
    \end{equation*}
    and where for the first inequality we have used that since
    $V\in\mc{V}$, we have $r_{u,v} \geq n^{-1}$ for all $u,v\in V$.

    Similarly, for \eqref{eq:thm4},
    \begin{equation*}
        \sum_{w\neq u}\lambdauc_{w,u} \leq C(\{u\}^c,\{u\}),
    \end{equation*}
    and
    \begin{align*}
        C(\{u\}^c,\{u\}) 
        & \leq \log\Big(1+(n-1){\textstyle \sum_{v\neq u} } \abs{h_{v,u}}^2\Big) \\
        & \leq \log(1+(n-1)^2n^\alpha) \\
        & \leq K\log(n).
    \end{align*}
\end{proof}

The next lemma bounds the maximal achievable sum rate across the
boundary out of the subsquares $V_{\ell,i}(n)$ for $\ell\in\{1,\ldots,
L(n)\}$, and $i\in\{1,\ldots,4^{\ell}\}$.

\begin{lemma}
    \label{thm:cutset}
    Under either fast or slow fading, for any $\alpha>2$, there exists
    $b(n)=O\big(\log^6(n)\big)$ such that for all
    $V(n)\in\mc{V}(n)$, $\lambdauc\in\Lambdauc(n)$, $\ell\in\{1,\ldots,
    L(n)\}$, and $i\in\{1,\ldots,4^{\ell}\}$, we have
    \begin{align*}
        \sum_{u\in V_{\ell,i}(n)}\sum_{w\notin V_{\ell,i}(n)}\lambdauc_{u,w} 
        \leq b(n) (4^{-\ell}n)^{2-\min\{3,\alpha\}/2}.
    \end{align*}
\end{lemma}
\begin{proof}
    As before, denote by $C(S_1,S_2)$ the MIMO capacity between nodes in
    $S_1$ and nodes in $S_2$. By the cut-set bound \cite[Theorem
    14.10.1]{cov},
    \begin{equation}
        \label{eq:cutset0a}
        \sum_{u\in V_{\ell,i}}\sum_{w\notin V_{\ell,i}}\lambdauc_{u,w} 
        \leq C(V_{\ell,i},V_{\ell,i}^c).
    \end{equation}
    Let
    \begin{equation*}
        \bm{H}_{S_1,S_2} \defeq [h_{u,v}]_{u\in S_1,v\in S_2}
    \end{equation*}
    be the matrix of channel gains between the nodes in $S_1$ and $S_2$.
    Under fast fading
    \begin{equation*}
        C(S_1,S_2)
        \defeq  \max_{\substack{\bm{Q}(\bm{H})\geq 0: \\ \E(q_{u,u})\leq P\ \forall u\in S_1}}
        \E\bigg( \log \det\big(\bm{I}+\bm{H}_{S_1,S_2}^\dagger\bm{Q}(\bm{H}) \bm{H}_{S_1,S_2}\big)\bigg),
    \end{equation*}
    and under slow fading
    \begin{equation*}
        C(S_1,S_2) 
        \defeq \max_{\substack{\bm{Q}\geq 0: \\ q_{u,u}\leq P\ \forall u\in S_1}}
        \log \det\big(\bm{I}+\bm{H}_{S_1,S_2}^\dagger\bm{Q} \bm{H}_{S_1,S_2}\big).
    \end{equation*}
    Denote by $\partial (V_{\ell,i}^c)$ the nodes in $V_{\ell,i}^c$ that are
    within distance one of the boundary between $A_{\ell,i}^c$ and
    $A_{\ell,i}$. Using the generalized Hadamard inequality yields
    that under either fast or slow fading
    \begin{equation}
        \label{eq:cutset0c}
        C(V_{\ell,i},V_{\ell,i}^c)
        \leq C(V_{\ell,i},\partial (V_{\ell,i}^c)) + 
        C(V_{\ell,i},V_{\ell,i}^c \setminus \partial (V_{\ell,i}^c)).
    \end{equation}

    We start by analyzing the first term in the sum in
    \eqref{eq:cutset0c}. Applying Hadamard's inequality again yields
    \begin{equation*}
        C(V_{\ell,i},\partial (V_{\ell,i}^c))
        \leq \sum_{v\in\partial (V_{\ell,i}^c)} C(V_{\ell,i},\{v\}).
    \end{equation*}
    Since $V\in\mc{V}$, we have 
    \begin{equation*}
        \card{\partial (V_{\ell,i}^c)} \leq 5 \log(n)(4^{-\ell}n)^{1/2}.
    \end{equation*}
    By the same analysis as in Lemma \ref{thm:cutset2}, we
    obtain
    \begin{equation*}
        C(V_{\ell,i},\{v\}) 
        \leq C(\{v\}^c,\{v\}) 
        \leq \frac{K}{5}\log(n)
    \end{equation*}
    for some constant $K$ (independent of $v$). Therefore
    \begin{align}
        \label{eq:cutset0d}
        C(V_{\ell,i},\partial (V_{\ell,i}^c)) 
        & \leq 5 \log(n)(4^{-\ell}n)^{1/2}\frac{K}{5}\log(n) \nonumber\\
        & \leq  K\log^2(n)(4^{-\ell}n)^{1/2}.
    \end{align}

    We now analyze the second term in the sum in
    \eqref{eq:cutset0c}.  
    The arguments of \cite[Lemma 12]{nie} (building on \cite[Theorem
    5.2]{ozg}) show that under either fast or slow fading there exists
    $\widetilde{K}>0$ such that for any $V\in\mc{V},\ell\in
    \{0,\ldots,L(n)\}$,
    \begin{equation}
        \label{eq:cutset0z}
        C(V_{\ell,i}, V_{\ell,i}^c\setminus \partial (V_{\ell,i}^c))
        \leq \widetilde{K} \log^3(n) \sum_{u\in V_{\ell,i}} 
        \sum_{v\in V_{\ell,i}^c \setminus \partial (V_{\ell,i}^c)} r_{u,v}^{-\alpha}.
    \end{equation}
    Moreover, using the same arguments as in \cite[Theorem
    5.2]{ozg} shows that there exists a constant $K'>0$ such that for adjacent
    squares (i.e., sharing a side) $A_{\ell,i},A_{\ell,j}$,
    \begin{equation}
        \label{eq:ptot}
        \sum_{u\in V_{\ell,i}}
        \sum_{v\in V_{\ell,j} \setminus \partial (V_{\ell,i}^c)} r_{u,v}^{-\alpha} \\
        \leq K' \log^3(n) (4^{-\ell}n)^{2-\min\{3,\alpha\}/2}.
    \end{equation}
    Consider now two diagonal squares (i.e., sharing a corner point)
    $A_{\ell,i},A_{\ell,j}$. Using a similar argument and suitably
    redefining $K'$ shows that \eqref{eq:ptot} holds for diagonal
    squares as well. 

    Using this, we now compute the summation in \eqref{eq:cutset0z}.
    Consider ``rings'' of squares around $A_{\ell,i}$. The first such
    ``ring'' contains the (at most) $8$ squares neighboring
    $A_{\ell,i}$. The next ``ring'' contains at most $16$ squares. In
    general, ``ring'' $k$ contains at most $8k$ squares. Let
    \begin{equation*}
        \{A_{\ell,j}\}_{j\in I_k}
    \end{equation*}
    be the squares in ``ring'' $k$. Then
    \begin{equation}
        \label{eq:ring1}
        \sum_{u\in V_{\ell,i}}
        \sum_{v\in V_{\ell,i}^c \setminus \partial (V_{\ell,i}^c) } r_{u,v}^{-\alpha} \\
        = \sum_{k\geq 1}\sum_{j\in I_k}\sum_{u\in V_{\ell,i}}
        \sum_{v\in V_{\ell,j} \setminus \partial (V_{\ell,i}^c) } r_{u,v}^{-\alpha}.
    \end{equation}
    By \eqref{eq:ptot},
    \begin{equation}
        \label{eq:ring2} 
        \sum_{j\in I_1}\sum_{u\in V_{\ell,i}}
        \sum_{v\in V_{\ell,j}\setminus \partial (V_{\ell,i}^c) } r_{u,v}^{-\alpha} \\
        \leq 8 K' \log^3(n) (4^{-\ell}n)^{2-\min\{3,\alpha\}/2}.
    \end{equation}
    Now note that for $k>1$ and $j\in I_k$, nodes 
    $u\in V_{\ell,i}$ and $v\in V_{\ell,j}$ are at least at distance
    $r_{u,v}\geq (k-1)(2^{-\ell}\sqrt{n})$. Moreover, since $V\in\mc{V}$,  
    each $\{V_{\ell,j}\}_{\ell,j}$ has cardinality at most
    $4^{-\ell+1}n$. Thus
    \begin{align}
        \label{eq:ring3}
        \sum_{k>1}\sum_{j\in I_k} \sum_{u\in V_{\ell,i}}
        \sum_{v\in V_{\ell,j} \setminus \partial (V_{\ell,i}^c)} r_{u,v}^{-\alpha} 
        & \leq \sum_{k>1} 8k \big(4^{-\ell+1}n\big)^2 
        \big((k-1)(2^{-\ell}\sqrt{n})\big)^{-\alpha} \nonumber\\
        & = 128 \big(4^{-\ell}n\big)^{2-\alpha/2}
        \sum_{k>1} k  (k-1)^{-\alpha} \nonumber\\
        & = K'' \big(4^{-\ell}n\big)^{2-\alpha/2},
    \end{align}
    for some $K''> 0$, and where we have used that $\alpha > 2$.
    Substituting \eqref{eq:ring2} and \eqref{eq:ring3} into
    \eqref{eq:ring1} yields
    \begin{equation*}
        \sum_{u\in V_{\ell,i}}
        \sum_{v\in V_{\ell,i}^c \setminus \partial (V_{\ell,i}^c)} r_{u,v}^{-\alpha} \\
        \leq 8 K' \log^3(n) (4^{-\ell}n)^{2-\min\{3,\alpha\}/2}
        +K'' \big(4^{-\ell}n\big)^{2-\alpha/2},
    \end{equation*}
    and hence by \eqref{eq:cutset0z}
    \begin{equation}
        \label{eq:sum}
        C(V_{\ell,i}, V_{\ell,i}^c\setminus \partial (V_{\ell,i}^c)) 
        \leq \widetilde{K} \log^3(n) 
        \Big(8 K' \log^3(n) (4^{-\ell}n)^{2-\min\{3,\alpha\}/2} 
        +K'' \big(4^{-\ell}n\big)^{2-\alpha/2}\Big).
    \end{equation}
    
    Combining \eqref{eq:cutset0a}, \eqref{eq:cutset0c},
    \eqref{eq:cutset0d}, and \eqref{eq:sum} shows that
    \begin{equation*}
        \sum_{u\in V_{\ell,i}}\sum_{v\notin V_{\ell,i}}\lambdauc_{u,v} 
        \leq b(n) (4^{-\ell}n)^{2-\min\{3,\alpha\}/2}.
    \end{equation*}
    for every $\ell\in\{1,\ldots,L(n)\}, i\in\{1,\ldots,4^\ell\}$, and
    under either fast or slow fading.
\end{proof}

The following lemma bounds the maximal achievable sum rate across the
boundary into the subsquares $V_{\ell,i}(n)$ for $\ell\in\{1,\ldots,
L(n)\}$, and $i\in\{1,\ldots,4^{\ell}\}$. Note that this lemma is only
valid for $\alpha>5$.

\begin{lemma}
    \label{thm:cutset3}
    Under either fast or slow fading, for any $\alpha>5$, there exists
    $b(n)=O\big(\log^3(n)\big)$ such that for all
    $V(n)\in\mc{V}(n)$, $\lambdauc\in\Lambdauc(n)$, $\ell\in\{1,\ldots,
    L(n)\}$, and $i\in\{1,\ldots,4^{\ell}\}$, we have
    \begin{equation*}
        \sum_{u\notin V_{\ell,i}(n)}\sum_{w\in V_{\ell,i}(n)}\lambdauc_{u,w} 
        \leq b(n) (4^{-\ell}n)^{1/2}.
    \end{equation*}
\end{lemma}
\begin{proof}
    By the cut-set bound \cite[Theorem 14.10.1]{cov},
    \begin{equation}
        \label{eq:cutset3_0}
        \sum_{u\notin V_{\ell,i}}\sum_{w\in V_{\ell,i}}\lambdauc_{u,w} 
        \leq C(V_{\ell,i}^c,V_{\ell,i}).
    \end{equation}
    Denote by $\partial V_{\ell,i}$ the nodes in $V_{\ell,i}$ that are
    within distance one of the boundary between $A_{\ell,i}^c$ and
    $A_{\ell,i}$. Applying
    the generalized Hadamard inequality as in Lemma
    \ref{thm:cutset}, we have under either fast or slow fading
    \begin{equation}
        \label{eq:cutset3_1}
        \begin{aligned}
            C(V_{\ell,i}^c,V_{\ell,i}) 
            & \leq C(V_{\ell,i}^c,\partial V_{\ell,i})+C(V_{\ell,i}^c, V_{\ell,i}\setminus\partial V_{\ell,i}) \\
            & \leq K\log^2(n)(4^{-\ell}n)^{1/2} +C(V_{\ell,i}^c, V_{\ell,i}\setminus\partial V_{\ell,i}),
        \end{aligned}
    \end{equation}
    for some positive constant $K$.

    For the second term in \eqref{eq:cutset3_1}, we have by slightly
    adapting the upper bound from Theorem 2.1 in \cite{jov}:
    \begin{equation*}
        C(V_{\ell,i}^c, V_{\ell,i}\setminus\partial V_{\ell,i})
        \leq \sum_{v\in V_{\ell,i}\setminus\partial V_{\ell,i}}
        \Big(\sum_{u \in V_{\ell,i}^c}r_{u,v}^{-\alpha/2}\Big)^2.
    \end{equation*}
    Now, consider $v\in V_{\ell,i}\setminus\partial V_{\ell,i}$ and let
    $d_v$ be the distance of $v$ from the closest node in
    $V_{\ell,i}^c$. Using $V\in\mc{V}$ and $\alpha>5$,
    \begin{equation*}
        \sum_{u\in V_{\ell,i}^c} r_{u,v}^{-\alpha/2}
        \leq \widetilde{K}\log(n) d_v^{2-\alpha/2},
    \end{equation*}
    for some positive constant $\widetilde{K}$, and hence 
    \begin{align*}
        C(V_{\ell,i}^c, V_{\ell,i}\setminus\partial V_{\ell,i})
        & \leq \sum_{v\in V_{\ell,i}\setminus\partial V_{\ell,i}}
        \widetilde{K}^2\log^2(n) d_v^{4-\alpha} \\
        & \leq K'\log^3(n)(4^{-\ell}n)^{1/2},
    \end{align*}
    for some positive constant $K'$.  Combined with \eqref{eq:cutset3_1} and
    \eqref{eq:cutset3_0}, this proves Lemma \ref{thm:cutset3}.
\end{proof}

\subsection{Achievability Lemmas}
\label{sec:aux_achievability}

In this section, we prove auxiliary achievability results. Recall that
a permutation traffic is a traffic pattern in which each node is source
and destination exactly once. Call the corresponding source-destination
pairing $\Pi\subset V(n)\times V(n)$ a \emph{permutation pairing}.  The
lemma below analyzes the performance achievable with either hierarchical
relaying (for $\alpha\in(2,3]$) or multi-hop communication (for $\alpha > 3$)
applied simultaneously to transmit permutation traffic in several
disjoint regions in the network. See Section \ref{sec:schemes_hr} for a
description of these communication schemes.

\begin{lemma}
    \label{thm:hr}
    Under fast fading, for any $\alpha>2$, there exists $b(n) \geq
    n^{-o(1)}$ such that for all $V(n)\in\mc{V}(n)$,
    $\ell\in\{0,\ldots,L(n)\}$, $i\in\{1,\ldots,4^\ell\}$, and
    permutation source-destination pairing $\Pi_i$ on $V_{\ell,i}(n)$,
    there exists $\lambdauc\in\Lambdauc(n)$ such that
    \begin{equation*}
        \min_{i\in\{1,\ldots,4^{\ell}\}}\min_{(u,w)\in \Pi_i} 
        \lambdauc_{u,w} 
        \geq b(n)(4^{-\ell}n)^{1-\min\{3,\alpha\}/2}.
    \end{equation*}
    The same statement holds with probability $1-o(1)$ as $n\to\infty$
    in the slow fading case.
\end{lemma}

Consider the source-destination pairing $\Pi\defeq\cup_i \Pi_i$ with
$\{\Pi_i\}_i$ as in Lemma \ref{thm:hr}. This is a permutation pairing, since
each $\Pi_i$ is a permutation pairing on $V_{\ell,i}(n)$ and since the
$\{V_{\ell,i}(n)\}_i$ are disjoint. Lemma \ref{thm:hr} states that every 
source-destination pair in $\Pi$ can communicate at a per-node rate of
at least $n^{-o(1)}(4^{-\ell}n)^{1-\min\{3,\alpha\}/2}$. Note that, due
to the locality of the traffic pattern, this can be considerably better
than the $n^{1-\min\{3,\alpha\}/2-o(1)}$ per-node rate achieved by
standard hierarchical relaying or multi-hop communication.

\begin{proof}
    We shall use either hierarchical relaying (for $\alpha\in(2,3]$) or
    multi-hop (for $\alpha>3$) to communicate within each square
    $V_{\ell,i}$. We operate every fourth of the $V_{\ell,i}$
    simultaneously, and show that the added interference due to this
    spatial re-use results only in a constant factor loss in rate.

    Consider first $\alpha\in(2,3]$ and fast fading. The squares
    $A_{\ell,i}$ at level $\ell$ have an area of 
    \begin{equation*}
        n_\ell \defeq 4^{-\ell}n. 
    \end{equation*}
    In order to be able to use hierarchical relaying within each of the
    $\{A_{\ell,i}\}_i$, it is sufficient to show that we can partition
    each $A_{\ell,i}$ into
    \begin{equation*}
        n_\ell^{\frac{1}{1+\log^{-1/3}(n_\ell)}}
    \end{equation*}
    squarelets, each of which contains a number of nodes proportional to
    the area (see Section \ref{sec:schemes_hr}). In other words, we
    partition $A$ into squarelets of size
    \begin{align*}
        n_\ell^{1-\frac{1}{1+\log^{-1/3}(n_\ell)}}
        & \geq n_{L(n)}^{\frac{\log^{-1/3}(n)}{1+\log^{-1/3}(n)}} \\
        & \geq n_{L(n)}^{\frac{1}{2}\log^{-1/3}(n)} \\
        & = 2^{\frac{1}{2}\log^{1/6}(n)} \\
        & \geq n4^{-\frac{1}{2}\log(n)\big(1-\tfrac{1}{2}\log^{-5/6}(n)\big)},
    \end{align*}
    where we have assumed, without loss of generality, that $n\geq 2$.
    Since $V\in\mc{V}$, all these squarelets contain a number of nodes
    proportional to their area, and hence this shows that all
    \begin{equation*}
        \{A_{i,\ell}\}_{\ell\in\{0,\ldots, L(n)\},i\in\{1,\ldots,4^\ell\}}
    \end{equation*}
    are simultaneously regular enough for hierarchical relaying to be
    successful under fast fading. This achieves a per-node rate of 
    \begin{equation}
        \label{eq:hr_step1}
        \lambdauc_{u,w} \geq n^{-o(1)}(4^{-\ell}n)^{1-\alpha/2}
    \end{equation}
    for any $(u,w)\in \Pi_i$ (see Section \ref{sec:schemes_hr}, or
    \cite[Theorem 1]{nie}).

    We now show that \eqref{eq:hr_step1} holds with high probability
    also under slow fading. For $V\in\mc{V}$ hierarchical relaying is
    successful under slow fading for all permutation traffic on $V$ with
    probability at least
    \begin{equation*}
        1-\exp\Big(-2^{K\log^{2/3}(n)}\Big)
    \end{equation*}
    for some constant $K$ (see again Section \ref{sec:schemes_hr}).
    Hence, hierarchical relaying is successful for all permutation
    traffic on $V_{\ell,i}$ with probability at least
    \begin{align*}
        1-\exp\Big(-2^{K\log^{2/3}(4^{-\ell}n)}\Big) 
        & \geq 1-\exp\Big(-2^{K\log^{2/3}(4^{-L(n)}n)}\Big) \\
        & = 1-\exp\Big(-2^{K\log^{1/3}(n)}\Big).
    \end{align*}
    And hence hierarchical relaying is successful under slow fading for
    all $\ell\in\{1,\ldots,L(n)\}$ and all permutation traffic on every
    $\{V_{\ell,i}\}_{i=1}^{4^\ell}$ with probability at least
    \begin{align*}
        1-L(n)4^{L(n)}\exp\Big(-2^{K\log^{1/3}(n)}\Big) 
        & \geq 1-n^2\exp\Big(-2^{K\log^{1/3}(n)}\Big) \\
        & \geq 1-o(1)
    \end{align*}
    as $n\to\infty$.

    We now argue that the additional interference from spatial re-use
    results only in a constant loss in rate. This follows from the same
    arguments as in the proof of \cite[Theorem 1]{nie} (with the
    appropriate modifications for slow fading as described there).
    Intuitively, this is the case since the interference from a square
    at distance $r$ is attenuated by a factor $r^{-\alpha}$, which,
    since $\alpha > 2$, is summable.  Hence the combined interference
    has power on the order of the receiver noise, resulting in only a
    constant factor loss in rate.

    For $\alpha > 3$, the argument is similar---instead of
    hierarchical relaying we now use multi-hop communication. For
    $V\in\mc{V}$ and under either fast or slow fading, this achieves a
    per-node rate of 
    \begin{equation}
        \label{eq:hr_step2}
        \lambdauc_{u,w} \geq n^{-o(1)}(4^{-\ell}n)^{-1/2}
    \end{equation}
    for any $(u,w)\in \Pi_i$.  Combining \eqref{eq:hr_step1} and
    \eqref{eq:hr_step2} yields the lemma.
\end{proof}

\section{Proof of Theorem \ref{thm:unicast}}
\label{sec:proof_unicast}

The proof of Theorem \ref{thm:unicast} relies on the construction of a
capacitated (noiseless, wireline) graph $G$ and linking its performance
under routing to the performance of the wireless network. This graph $G
= (V_G,E_G)$ is constructed as follows. $G$ is a full tree (i.e., all
its leaf nodes are on the same level). $G$ has $n$ leaves, each of them
representing an element of $V(n)$. To simplify notation, we assume that
$V(n)\subset V_G$, so that the leaves of $G$ are exactly the elements of
$V(n)\subset V_G$. Whenever the distinction is relevant, we use $u,v$
for nodes in $V(n)\subset V_G$ and $\mu,\nu$ for nodes in $V_G\setminus
V(n)$ in the following. The internal nodes of $G$ correspond to
$V_{\ell,i}(n)$ for all $\ell \in\{0,\ldots,L(n)\}$, $i \in\{1,\ldots,
4^{\ell}\}$, with hierarchy induced by the one on $A(n)$. In particular,
let $\mu$ and $\nu$ be internal nodes in $V_G$ and let $V_{\ell,i}(n)$
and $V_{\ell+1,j}(n)$ be the corresponding subsets of $V(n)$. Then $\nu$
is a child node of $\mu$ if $V_{\ell+1,j}(n)\subset V_{\ell,i}(n)$.

In the following, we will assume $V\in\mc{V}$, which holds with
probability $1-o(1)$ as $n\to\infty$ by Lemma \ref{thm:mcv}.  With this
assumption, nodes in $V_G$ at level $\ell<L(n)$ have $4$ children each,
nodes in $V_G$ at level $\ell=L(n)$ have between $4^{-L(n)-1}n$ and
$4^{-L(n)+1}n$ children, and nodes in $V_G$ at level $\ell=L(n)+1$ are
the leaves of the tree (see Figure \ref{fig:graph} below and Figure
\ref{fig:grid_graph} in Section \ref{sec:schemes_unicast}). 
\begin{figure}[!ht]
    \begin{center}
        \scalebox{0.8}{
        \input{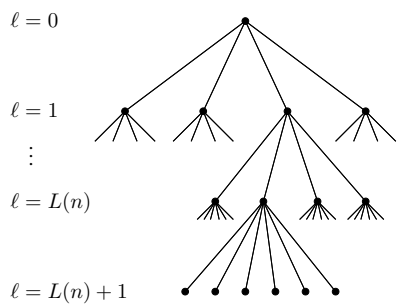}
        }
    \end{center}

    \caption{Communication graph $G$ constructed in the proof of Theorem
    \ref{thm:unicast}.  Nodes on levels $\ell\in\{0,\ldots,L(n)-1\}$
    have four children each, nodes on level $\ell=L(n)$ have
    $\Theta\big(n^{\log^{-1/2}(n)}\big)$ children each.  The total number of
    leaf nodes is $n$, one representing each node in the wireless
    network $V(n)$. An internal node in $G$ at level
    $\ell\in\{0,\ldots, L(n)\}$ represents the collection of nodes in
    $V_{\ell,i}(n)$ for some $i$.
    }

    \label{fig:graph}
\end{figure}

For $\mu\in V_G$, denote by $\mc{L}(\mu)$ the leaf nodes of the subtree
of $G$ rooted at $\mu$. Note that, by construction of the graph $G$,
$\mc{L}(\mu)= V_{\ell,i}(n)$ for some $\ell$ and $i$.  To understand the
relation between $V_G$ and $V(n)$, we define the \emph{representative}
$\mc{R}:V_G\to 2^{V(n)}$ of $\mu$ as follows.  For a leaf node $u\in
V(n)\subset V_G$ of $G$, let
\begin{equation*}
    \mc{R}(u) \defeq \{u\}.
\end{equation*}
For $\mu\in V_G$ at level $L(n)$, choose
$\mc{R}(\mu) \subset \mc{L}(\mu)\subset V(n)$
such that
\begin{equation*}
    \card{\mc{R}(\mu)}
    = 4^{-L(n)-1}n.
\end{equation*}
This is possible since $V(n)\in\mc{V}(n)$ by assumption. Finally, for
$\mu\in V_G$ at level $\ell < L(n)$, and with children
$\{\nu_i\}_{j=1}^4$, let
\begin{equation*}
    \mc{R}(\mu) \defeq \bigcup_{j=1}^4 \mc{R}(\nu_j).
\end{equation*}

We now define an edge capacity $c_{\mu,\nu}$ for each edge
$(\mu,\nu)\in E_G$.  If $\mu$ is a leaf of $G$ and $\nu$ its
parent, set 
\begin{equation}
    \label{eq:leaf}
    c_{\mu,\nu}
    = c_{\nu,\mu}
    \defeq 1.
\end{equation}
If $\mu$ is an internal  node at level $\ell$ in $G$ and $\nu$ its
parent, then set
\begin{equation}
    \label{eq:internal}
    c_{\mu,\nu}
    = c_{\nu,\mu}
    \defeq (4^{-\ell}n)^{2-\min\{3,\alpha\}/2}.
\end{equation}

Having chosen edge capacities on $G$, we can now define the set
$\Lambdauc_G(n)\subset \Rp^{n\times n}$ of feasible unicast traffic
matrices between leaf nodes of $G$. In other words,
$\lambdauc\in\Lambdauc_G(n)$ if messages at the leaf nodes of $G$ can be
routed to their destinations (which are also leaf nodes) over $G$ at
rates $\lambdauc$ while respecting the capacity constraints on the edges
of $G$. Define
\begin{equation*}
    \Lambdabuc_G(n) \defeq \Lambdauc_G(n)\cap\Buc(n).
\end{equation*}

We first prove the achievability part of Theorem \ref{thm:unicast}. The
next lemma shows that if traffic can be routed over the tree $G$ then
approximately the same traffic can be transmitted reliably over the
wireless network.
\begin{lemma}
    \label{thm:tree_equivalence1_unicast}
    Under fast fading, for any $\alpha>2$, there exists $b(n) \geq
    n^{-o(1)}$ such that for any $V(n)\in\mc{V}(n)$,
    \begin{align*}
        b(n)\Lambdauc_G(n) \subset\Lambdauc(n).
    \end{align*}
    The same statement holds for slow fading with probability $1-o(1)$
    as $n\to\infty$.
\end{lemma}
\begin{proof}
    Assume $\lambdauc\in\Lambdauc_G$, i.e., traffic can be routed
    between the leaf nodes of $G$ at a rate $\lambdauc$, we need to show
    that $n^{-o(1)}\lambdauc\in\Lambdauc$ (i.e., almost the same flow
    can be reliably transmitted over the wireless network). We use the
    three-layer communication architecture introduced in Section
    \ref{sec:schemes_unicast} to establish this result.

    Recall the three layers of this architecture: the routing,
    cooperation, and physical layers. The layers of this communication
    scheme operate as follows. In the routing layer, we treat the
    wireless network as the graph $G$ and route the messages between
    nodes over the edges of $G$. The cooperation layer provides this
    tree abstraction to the routing layer by distributing and
    concentrating messages over subsets of the wireless networks. The
    physical layer implements this distribution and concentration of
    messages by dealing with interference and noise.

    Consider first the routing layer, and assume that the tree
    abstraction $G$ can be implemented in the wireless network with only
    a $n^{-o(1)}$ factor loss. Since $\lambdauc\in\Lambdauc_G$ by
    assumption, we then know that the routing layer will be able to
    reliably transmit messages at rates $n^{-o(1)}\lambdauc$ over the
    wireless network. We now show that the tree abstraction can indeed
    be implemented with a factor $n^{-o(1)}$ loss in the wireless
    network.

    This tree abstraction is provided to the routing layer by the
    cooperation layer. We will show that the operation of the
    cooperation layer satisfies the following invariance property: If a
    message is located at a node $\mu\in G$ in the routing layer, then
    the same message is evenly distributed over all nodes in
    $\mc{R}(\mu)$ in the wireless network. In other words, all nodes
    $u\in\mc{R}(\mu)\subset V$ contain a distinct part of length
    $1/\card{\mc{R}(\mu)}$ of the message. 

    Consider first a leaf node $u\in V\subset V_G$ in $G$, and assume
    the routing layer calls upon the cooperation layer to send a message
    to its parent $\nu\in V_G$ in $G$. Note first that $u$ is also an
    element of $V$, and it has access to the entire message to be sent
    over $G$. Since for leaf nodes $\mc{R}(u)=\{u\}$, this shows that
    the invariance property is satisfied at $u$. The message is split at
    $u$ into $\card{\mc{R}(\nu)}$ parts of equal length, and one part is
    sent to each node in $\mc{R}(\nu)$ over the wireless network. In
    other words, we distribute the message over the wireless network by
    a factor of $\card{\mc{R}(\nu)}$. Hence the invariance property is
    also satisfied at $\nu$. 

    Consider now an internal node $\mu\in V_G$, and assume the routing
    layer calls upon the cooperation layer to send a message to its
    parent node $\nu\in V_G$.  Note that since all traffic in $G$
    originates at the leaf nodes of $G$ (which are the actual nodes in
    the wireless network), a message at $\mu$ had to traverse all levels
    below $\mu$ in the tree $G$. We assume that the invariance property
    holds up to level $\mu$ in the tree, and show that it is then also
    satisfied at level $\nu$. By the induction hypothesis, each node
    $u\in\mc{R}(\mu)$ has access to a distinct part of length
    $1/\card{\mc{R}(\mu)}$. Each such node $u$ splits its message part
    into four distinct parts of equal length. Node $u$ keeps one part
    for itself, and sends the other three parts to nodes in
    $\mc{R}(\nu)$.  Since $\card{\mc{R}(\nu)}=4\card{\mc{R}(\mu)}$, this
    can be performed such that each node in $\mc{R}(\nu)$ obtains
    exactly one message part.  In other words, we distribute the message
    by a factor four over the wireless network, and the invariance
    property is satisfied at $\nu\in V_G$. 

    Operation along edges down the tree (i.e., towards the leaf nodes)
    is similar, but instead of distributing messages, we now concentrate
    them over the wireless network. To route a message from a node
    $\mu\in V_G$ with internal children $\{\nu_j\}_{j=1}^4$ to one of
    them (say $\nu_1$) in the routing layer, the cooperation layer sends
    the message parts from each $\{\mc{R}(\nu_j)\}_{j=2}^4$ to a
    corresponding node in $\mc{R}(\nu_1)$ and combines them there.  In
    other words, we concentrate the message by a factor four over the
    wireless network. 

    To route a message to a leaf node $u\in V\subset
    V_G$ from its parent $\nu$ in $G$ in the routing layer, the
    cooperation layer sends the corresponding message parts at each node
    $\mc{R}(\nu)$ to $u$ over the wireless network. Thus, again we
    concentrate the message over the network, but this time by a factor
    of $\card{\mc{R}(\nu)}$. Both these operations along edges down the
    tree preserve the invariance property. This shows that the
    invariance property is preserved by all operations 
    induced by the routing layer in the cooperation layer.

    Finally, to actually implement this distribution and concentration
    of messages, the cooperation layer calls upon the physical layer.
    Note that at the routing layer, all edges of the tree can be routed
    over simultaneously. Therefore, the cooperation layer can
    potentially call the physical layer to perform distribution and
    concentration of messages over all sets $\{\mc{R}(\mu)\}_{\mu\in
    V_G}$ simultaneously. The function of the physical layer is to
    schedule all these operations and to deal with the resulting
    interference as well as with channel noise. 

    This scheduling is done as follows. First, the physical layer time
    shares between communication up the tree and communication down the
    tree (i.e., between distribution and concentration of messages).
    This results in a loss of a factor $1/2$ in rate. The physical layer
    further time shares between all the $L(n)+1$ internal levels of the
    tree, resulting in a further $\frac{1}{L(n)+1}$ factor loss in rate.
    Hence, the total rate loss by this time sharing is
    \begin{equation}
        \label{eq:sharing}
        \frac{1}{2(L(n)+1)}.
    \end{equation}

    Consider now the operations within some level
    $\ell\in{1,\ldots,L(n)}$ in the tree (i.e., for edge $(\mu,\nu)$ on
    this level, neither $\mu$ nor $\nu$ is a leaf node). We show that
    the rate at which the physical layer implements the edge $(\mu,\nu)$
    is equal to $n^{-o(1)}c_{\mu,\nu}$, i.e., only a small factor less
    than the capacity of the edge $(\mu,\nu)$ in the tree $G$. Note
    first that the distribution or concentration of traffic induced by
    the cooperation layer to implement one edge $e$ at level $\ell$
    (i.e., between node levels $\ell$ and $\ell-1$) is
    restricted to $V_{\ell-1,i}$ for some $i=i(e)$. We can thus
    partition the edges at level $\ell$ into $\{E_G^{j}\}_{j=1}^{4}$
    such that the four sets
    \begin{equation*}
        \bigcup_{e\in E_G^j}V_{\ell-1,i(e)}
    \end{equation*}
    of nodes are disjoint. Time sharing between these four sets yields
    an additional loss of a factor $1/4$ in rate. Fix one such value of
    $j$, and consider the operations induced by the cooperation layer in
    the set corresponding to $j$. We consider communication up the tree
    (i.e., distribution of messages), the analysis for communication
    down the tree is similar. For a particular edge $(\mu,\nu)\in E_G^j$
    with $\nu$ the parent of $\mu$, each node $u\in\mc{R}(\mu)$ has
    split its message part into four parts, three of which need to be
    sent to the nodes in $\mc{R}(\nu)\setminus\mc{R}(\mu)$. Moreover,
    this assignment of destination nodes in
    $\mc{R}(\nu)\setminus\mc{R}(\mu)$ to $u$ is performed such that no
    node in $\mc{R}(\nu)\setminus\mc{R}(\mu)$ is destination more than
    once. In other word, each node in $\mc{R}(\mu)$ is source exactly
    three times and each node in $\mc{R}(\nu)\setminus\mc{R}(\mu)$ is
    destination exactly once. This can be written as three
    source-destination pairings $\{\Pi_{i(\mu,\nu)}^k\}_{k=1}^{3}$, on
    $V_{\ell-1, i(\mu,\nu)}$.  Moreover, each such $\Pi_{i(\mu,\nu)}^k$
    can be understood as a subset of a permutation source-destination
    pairing. We time share between the three values
    of $k$ (yielding an additional loss of a factor $1/3$ in rate).
    Now, for each value of $k$, Lemma \ref{thm:hr} shows that by using
    either hierarchical relaying (for $\alpha\in(2,3]$) or multi-hop
    communication for ($\alpha>3$), we can communicate according to
    $\{\Pi_{i(e)}^k\}_{e\in E_G^j}$ at a per-node rate of 
    \begin{equation*}
        n^{-o(1)}(4^{-\ell-1}n)^{1-\min\{3,\alpha\}/2}
    \end{equation*}
    under fast fading, and with probability\footnote{Note that Lemma
    \ref{thm:hr} actually shows that all permutation traffic for every
    value of $\ell$ can be transmitted with high probability under slow
    fading. In other words, with high probability all levels of $G$ can
    be implemented successfully under slow fading.} $1-o(1)$ also under
    slow fading. Since $\mc{R}(\mu)$ contains $4^{-\ell-1}n$ nodes, and
    accounting for the loss \eqref{eq:sharing} for time sharing between
    the levels in $G$ and the additional loss of factors $1/4$ and $1/3$
    for time sharing between $j$ and $k$, the physical layer implements
    an edge capacity for $e$ at level $\ell$ of
    \begin{equation*}
        \frac{1}{2(L(n)+1)}\cdot\frac{1}{4}\cdot\frac{1}{3}\cdot
        4^{-\ell-1}n\cdot n^{-o(1)}(4^{-\ell-1}n)^{1-\min\{3,\alpha\}/2}
        = n^{-o(1)}(4^{-\ell}n)^{2-\min\{3,\alpha\}/2} \\
        = n^{-o(1)}c_e.
    \end{equation*}

    Consider now the operations within level $\ell=L(n)+1$ in the tree
    (i.e., for edge $(u,\nu)$ on this level, $u$ is a leaf node). We
    show that the rate at which the physical layer implements the edge
    $(u,\nu)$ is equal to $n^{-o(1)}c_{u,\nu}$. We again consider only
    communication up the tree (i.e., distribution of messages in the
    cooperation layer), communication down the tree is performed in a
    similar manner. The traffic induced by the cooperation layer at
    level $L(n)+1$ is within the sets $V_{L(n),i}$ for
    $i=\{1,\ldots,4^{L(n)}\}$. 
    Consider now communication within one $V_{L(n),i}$, and
    assume without loss of generality that in the routing layer every
    node $u\in V_{L(n),i}$ needs to send traffic along the
    edge $(u,\nu)$. In the physical layer, we need to distribute a
    $1/\card{\mc{R}(\nu)}$ fraction of this traffic from each node $u
    \in V_{L(n),i}$ to every node in $\mc{R}(\nu)\subset
    V_{L(n),i}$. This can be expressed as
    $\card{V_{L(n),i}}$ source-destination pairings, and we
    time share between them. Accounting for the fact that only
    $1/\card{\mc{R}(\nu)}$ of traffic needs to be sent according to each
    pairing and since $V\in\mc{V}$, this results in a time sharing loss
    of at most a factor
    \begin{equation*}
        \frac{\card{\mc{R}(\nu)}}{\card{V_{L(n),i}}}
        \leq \frac{1}{16}.
    \end{equation*}
    Now, using Lemma \ref{thm:hr}, all these source-destination
    pairings in all subsquares $\{V_{L(n),i}\}$ can be
    implemented simultaneously at a per node rate of 
    \begin{equation*}
        n^{-o(1)}(4^{-L(n)}n)^{1-\min\{3,\alpha\}/2} 
        \geq n^{-o(1)}(n^{\log^{-1/2}(n)})^{-1/2} 
        \geq n^{-o(1)}.
    \end{equation*}
    Accounting for the loss \eqref{eq:sharing} for time
    sharing between the levels in $G$, the additional factor $1/16$
    loss for time sharing within each $V_{L(n),i}$, the physical
    layer implements an edge capacity for $e$ at level $\ell=L(n)+1$ of
    \begin{equation*}
        \frac{1}{2(L(n)+1)}\cdot\frac{1}{16}\cdot n^{-o(1)}
        = n^{-o(1)}
        = n^{-o(1)}c_e,
    \end{equation*}
    under either fast or slow fading.
 
    Together, this shows that the physical and cooperation layers
    provide the tree abstraction $G$ to the routing layer with edge
    capacities of only a factor $n^{-o(1)}$ loss. Hence, if messages
    can be routed at rates $\lambdauc$ between the leaf nodes of $G$,
    then messages can be reliably transmitted over the wireless network
    at rates $n^{-o(1)}\lambdauc$. Hence
    \begin{equation*}
        \lambdauc \in \Lambdauc_G \ \Rightarrow \ 
        n^{-o(1)}\lambdauc\in\Lambdauc.
    \end{equation*}
    Noting that the $n^{-o(1)}$ factor
    is uniform in $\lambdauc$, this shows that
    \begin{equation*}
        n^{-o(1)}\Lambdauc_G \subset\Lambdauc.
    \end{equation*}
\end{proof}

We have seen that the unicast capacity region $\Lambdauc_G(n)$ of the
graph $G$ under routing is (appropriately scaled) an inner bound to the
unicast capacity region $\Lambdauc(n)$ of the wireless network. Taking
the intersection with the set of balanced traffic matrices $\Buc(n)$
yields that the same holds for $\Lambdabuc_G(n)$ and $\Lambdabuc(n)$.
The next lemma shows that $(\gamma(n)+1)\Lambdabuc_G(n)$ (with
$\gamma(n)=n^{o(1)}$ as in the definition of $\Buc(n)$ in
\eqref{eq:unicast_balanced}) is an outer bound to the approximate
unicast capacity region $\hLambdabuc(n)$ of the wireless network as
defined in \eqref{eq:approx_unicast} Combining Lemmas \ref{thm:mcv},
\ref{thm:tree_equivalence1_unicast}, and \ref{thm:tree_approx} below,
yields that with high probability
\begin{equation*}
    n^{-o(1)}\hLambdabuc(n) 
    \subset n^{-o(1)}\Lambdabuc_G(n)
    \subset \Lambdabuc(n),
\end{equation*}
proving the achievability part of Theorem \ref{thm:unicast}.

\begin{lemma}
    \label{thm:tree_approx}
    For any $\alpha>2$ and any $V(n)\in\mc{V}(n)$,
    \begin{equation*}
        \hLambdabuc(n) \subset (\gamma(n)+1)\Lambdabuc_G(n),
    \end{equation*}
    where $\gamma(n)=n^{o(1)}$ is the factor in the definition of $\Buc(n)$ in
    \eqref{eq:unicast_balanced}.
\end{lemma}
\begin{proof}
    We first relate the total traffic across an edge $e$ in the graph
    $G$ to the total traffic across a cut $V_{\ell,i}$ for some $\ell$
    and $i$.

    Consider an edge $e=(\mu,\nu)\in E_G$, and assume first that $e$ connects
    nodes at level $\ell$ and $\ell-1$ in the tree with $\ell\geq L(n)$.
    We slight abuse of notation, set
    \begin{equation*}
        c_e \defeq c_{\mu,\nu}.
    \end{equation*}
    Note first that by \eqref{eq:internal} we have
    \begin{equation}
        \label{eq:tree_approx1}
        c_e = (4^{-\ell}n)^{2-\min\{3,\alpha\}/2}.
    \end{equation}
    Moreover, since $G$ is a tree, removing the edge $e$ from $E_G$
    separates the tree into two connected components, say
    $S_1,S_2\subset V_G$.  Consider now the leaf nodes in $S_1$. By the
    construction of the tree structure of $G$, these leaf nodes are
    either equal to $V_{\ell,i}$ or $V_{\ell,i}^c$ for some
    $i\in\{1,\ldots,4^\ell\}$.  Assume without loss of generality that
    they are equal to $V_{\ell,i}$. Then $V_{\ell,i}^c$ are the leaf
    nodes in $S_2$. Now since traffic is only assumed to be between leaf
    nodes of $G$, the total traffic demand between $S_1$ and $S_2$ is
    equal to
    \begin{equation}
        \label{eq:tree_approx2}
        \sum_{u\in V_{\ell,i}}\sum_{w\notin V_{\ell,i}^c}
        (\lambdauc_{u,w}+\lambdauc_{w,u}).
    \end{equation}
    By the tree structure of $G$, all this traffic has to be
    routed over edge $e$. 

    Consider now an edge $e$ connecting a node at level $L(n)+1$ and $L(n)$,
    i.e., a leaf node $u$ to its parent $\nu$. Then, by \eqref{eq:leaf},
    \begin{equation}
        \label{eq:tree_approx3}
        c_e = 1. 
    \end{equation}
    The total traffic crossing the edge $e$ is equal to
    \begin{equation}
        \label{eq:tree_approx4}
        \sum_{w\neq u} (\lambdauc_{u,w}+\lambdauc_{w,u}).
    \end{equation}

    We now show that
    \begin{equation}
        \label{eq:tree_approx7}
        \hLambdabuc \subset (\gamma(n)+1)\Lambdabuc_G.
    \end{equation}
    Assume $\lambdauc\in\hLambdabuc$, then
    \begin{equation*}
        \sum_{u\in V_{\ell,i}}\sum_{w\notin V_{\ell,i}}
        \lambdauc_{u,w} \leq (4^{-\ell}n)^{2-\min\{3,\alpha\}}
    \end{equation*}
    for all $\ell\in\{1,\ldots,L(n)\},i\in\{1,\ldots,4^\ell\}$, and
    \begin{equation*}
        \sum_{w\neq u} (\lambdauc_{u,w}+\lambdauc_{w,u}) \leq 1
    \end{equation*}
    for all $u\in V$. Since $\lambdauc$ is balanced, this implies that
    \begin{equation*}
        \frac{1}{\gamma(n)+1}
        \sum_{u\in V_{\ell,i}}\sum_{w\notin V_{\ell,i}} (\lambdauc_{u,w}+\lambdauc_{w,u}) 
        \leq (4^{-\ell}n)^{2-\min\{3,\alpha\}} 
    \end{equation*}
    for $\ell\leq L(n)$.  By \eqref{eq:tree_approx1},
    \eqref{eq:tree_approx2}, \eqref{eq:tree_approx3},
    \eqref{eq:tree_approx4}, we obtain that the traffic demand across
    each edge $e$ of the graph $G$ is less than $\gamma(n)+1$ times its
    capacity $c_e$. Therefore, using that $G$ is a tree,
    $\frac{1}{\gamma(n)+1}\lambdauc$ can be routed over $G$, i.e.,
    $\lambdauc\in(\gamma(n)+1)\Lambdabuc_G$. This proves
    \eqref{eq:tree_approx7}.
\end{proof}

We now turn to the converse part of Theorem \ref{thm:unicast}. The next
lemma shows that $\hLambdauc(n)$ (appropriately scaled) is an outer
bound to the unicast capacity region $\Lambdauc(n)$ of the wireless
network. Taking the intersection with the collection of balanced traffic
matrices $\Buc(n)$ and combining with Lemma \ref{thm:mcv}, this shows
that with high probability
\begin{equation*}
    \Lambdabuc(n) \subset O(\log^6(n))\hLambdabuc(n),
\end{equation*}
proving the converse part of Theorem \ref{thm:unicast}.

\begin{lemma}
    \label{thm:tree_equivalence2_unicast}
    Under either fast or slow fading, for any $\alpha>2$, there exists
    $b(n) = O(\log^6(n))$ such that for any $V(n)\in\mc{V}(n)$,
    \begin{align*}
        \Lambdauc(n) \subset b(n)\hLambdauc(n).
    \end{align*}
\end{lemma}
\begin{proof}
    Assume $\lambdauc\in\Lambdauc$. By Lemma~\ref{thm:cutset},
    we have for any $\ell\in\{1,\ldots, L(n)\}$ and $i\in\{1,\ldots,
    4^{\ell}\}$,
    \begin{align}
        \label{eq:step1a}
        \sum_{u\in V_{\ell,i}}\sum_{w\notin V_{\ell,i}}\lambdauc_{u,w}
        \leq K \log^6(n) (4^{-\ell}n)^{2-\min\{3,\alpha\}/2}
    \end{align}
    for some constant $K$ not depending on $\lambdauc$.

    Consider now $u\in V$. Lemma \ref{thm:cutset2} shows that
    \begin{align*}
        \sum_{w\neq u} \lambdauc_{u,w}
        & \leq \widetilde{K}\log(n), \\
        \sum_{w\neq u} \lambdauc_{w,u}
        & \leq \widetilde{K}\log(n),
    \end{align*}
    with constant $\widetilde{K}$ not depending on $\lambdauc$,
    and therefore,
    \begin{equation}
        \label{eq:step1c}
        \sum_{w\neq u} (\lambdauc_{u,w}+\lambdauc_{w,u})
        \leq 2\widetilde{K}\log(n).
    \end{equation}

    Combining \eqref{eq:step1a} and \eqref{eq:step1c} proves that
    there exists $b(n) = O(\log^6(n))$ such that
    $\lambdauc\in\Lambdauc$ implies $\lambdauc\in b(n)\hLambdauc$,
    proving the lemma.
\end{proof}

\section{Proof of Theorem \ref{thm:multicast}}
\label{sec:proof_multicast}

Consider again the tree graph $G = (V_G,E_G)$ with leaf nodes $V(n)\subset
V_G$ constructed in Section \ref{sec:proof_unicast}. As before, we
consider traffic between leaf nodes of $G$. In particular, any multicast
traffic matrix $\lambdamc\in\Rp^{n\times 2^n}$ for the wireless network
is also a multicast traffic matrix for the graph $G$. Denote by
$\Lambdamc_G(n)\subset \Rp^{n\times 2^n}$ the set of feasible (under
routing) multicast traffic matrices between leaf nodes of $G$, and set
\begin{equation*}
    \Lambdabmc_G(n) \defeq \Lambdamc(n)\cap\Bmc(n).
\end{equation*}

The next lemma shows that if multicast traffic can be routed over
$G$ then approximately the same multicast traffic can be transmitted
reliably over the wireless network. Taking the intersection with
$\Bmc(n)$ implies that the same result holds also for balanced traffic.
\begin{lemma}
    \label{thm:tree_equivalence_multicast}
    Under fast fading, for any $\alpha>2$, there exists $b(n) \geq
    n^{-o(1)}$ such that for all $V(n)\in\mc{V}(n)$,
    \begin{align*}
        b(n)\Lambdamc_G(n) \subset\Lambdamc(n).
    \end{align*}
    The same statement holds under slow fading with probability
    $1-o(1)$ as $n\to\infty$.
\end{lemma}
\begin{proof}
    The proof follows using the same construction as in Lemma
    \ref{thm:tree_equivalence1_unicast}.
\end{proof}

We now show that, since $G$ is a tree graph, $\hLambdabmc(n)$ is an inner
bound (up to a factor $\gamma(n)+1$) to the the multicast capacity region
$\Lambdabmc_G(n)$.  The fact that $G$ is a tree is critical for this
result to hold.
\begin{lemma}
    \label{thm:tree_multicast}
    For any $\alpha>2$, 
    \begin{equation*}
        \hLambdabmc(n) \subset (\gamma(n)+1)\Lambdabmc_G(n).
    \end{equation*}
    where $\gamma(n)=n^{o(1)}$ is the factor in the definition of $\Bmc(n)$ in
    \eqref{eq:multicast_balanced}.
\end{lemma}
\begin{proof}
    Assume $\lambdamc\in\Bmc\setminus\Lambdabmc_G$.
    Since $G$ is a tree, there is only one way to route multicast
    traffic from $u$ to $W$, namely along the subtree $G(\{u\}\cup W)$
    induced by $\{u\}\cup W$ (i.e., the smallest subtree of $G$ that
    covers $\{u\}\cup W$). Hence for any edge $e\in E_G$, the traffic
    $d_{\lambdamc}(e)$ that needs to be routed over $e$ is equal to
    \begin{equation*}
        d_{\lambdamc}(e) 
        = \sum_{\substack{u\in V, W\subset V:\\ e\in E_{G(\{u\}\cup W)}}}
        \lambdamc_{u,W}.
    \end{equation*}

    Now, since $\lambdamc\in\Bmc\setminus\Lambdabmc_G$, there exists
    $e\in E_G$ such that
    \begin{equation}
        \label{eq:tree_multicast1}
        d_{\lambdamc}(e)> c_e.
    \end{equation}
    Let $\ell$ be the level of this edge $e$ in $G$. We have
    \begin{equation}
        \label{eq:tree_multicast3}
        c_e = 
        \begin{cases}
            \big(4^{-\ell}n\big)^{2-\min\{3,\alpha\}/2} & \text{if $\ell \leq L(n)$}, \\
            1 & \text{else}.
        \end{cases}
    \end{equation}

    Assume first that
    $\ell\leq L(n)$ and let $i$ be such that the removal of the edge $e$ in $G$
    disconnects the leave nodes in $V_{\ell,i}$ from the ones in
    $V_{\ell,i}^c$. Then we have
    \begin{equation}
        \label{eq:tree_multicast2}
        d_{\lambdamc}(e)  
        = \sum_{u\in V_{\ell,i}}
        \sum_{\substack{W\subset V:\\ W\setminus V_{\ell,i}\neq\emptyset}}\lambdamc_{u,W}
        +\sum_{u\notin V_{\ell,i}(n)}
        \sum_{\substack{W\subset V:\\ W\cap V_{\ell,i}\neq\emptyset}}\lambdamc_{u,W}.
    \end{equation}

    Assume then that $\ell = L(n)+1$, and assume $e$ separates the leaf
    node $u$ from $\{u\}^c$ in $G$. Then
    \begin{equation*}
        \label{eq:tree_multicast5}
        d_{\lambdamc}(e) 
        = \sum_{\substack{W\subset V: \\ W\setminus\{u\}\neq\emptyset}}
        \lambdamc_{u,W}
        + \sum_{\tilde{u}\neq u}\sum_{\substack{W\subset V:\\ u\in W}}
        \lambdamc_{\tilde{u},W}.
    \end{equation*}

    If $\ell = L(n)+1$, then \eqref{eq:tree_multicast1}, \eqref{eq:tree_multicast3}, and
    \eqref{eq:tree_multicast2} imply that $\lambdamc\notin\hLambdabmc$
    and therefore $\lambdamc\notin\frac{1}{\gamma(n)+1}\hLambdabmc$.
    If $\ell \leq L(n)$ then, since $\lambdamc$ is $\gamma(n)$-balanced, we have
    \begin{equation}
        \label{eq:tree_multicast4}
        \sum_{u\in V_{\ell,i}}
        \sum_{\substack{W\subset V:\\ W\setminus V_{\ell,i}\neq\emptyset}}\lambdamc_{u,W}
        +\sum_{u\notin V_{\ell,i}(n)}
        \sum_{\substack{W\subset V:\\ W\cap V_{\ell,i}\neq\emptyset}}\lambdamc_{u,W}  
        \leq (\gamma(n)+1) \sum_{u\in V_{\ell,i}}
        \sum_{\substack{W\subset V:\\ W\setminus V_{\ell,i}\neq\emptyset}}\lambdamc_{u,W}.
    \end{equation}
    Combining \eqref{eq:tree_multicast1}, \eqref{eq:tree_multicast3},
    \eqref{eq:tree_multicast2}, and \eqref{eq:tree_multicast4} shows
    that $\lambdamc\notin\frac{1}{\gamma(n)+1}\hLambdabmc$ for $\ell\leq
    L(n)$ as well. 
    
    Hence, we have shown that $\lambdamc\in\Bmc\setminus\Lambdabmc_G$
    implies $\lambdamc\notin\frac{1}{\gamma(n)+1}\hLambdabmc$, proving
    the lemma.
\end{proof}

Combining Lemmas \ref{thm:tree_equivalence_multicast}, and
\ref{thm:tree_multicast}, and \ref{thm:mcv} shows that, with probability
$1-o(1)$ as $n\to\infty$,
\begin{equation*}
    n^{-o(1)}\hLambdabmc(n) 
    \subset n^{-o(1)}\Lambdabmc_G(n)
    \subset\Lambdabmc(n),
\end{equation*}
proving the inner bound in Theorem \ref{thm:multicast}.

We now turn to the proof of the outer bound to $\Lambdamc(n)$. The next
lemma combined with Lemma \ref{thm:mcv}, and taking the intersection
with $\Bmc(n)$, proves the outer bound in Theorem \ref{thm:multicast}. 
\begin{lemma}
    \label{thm:multicast_outer}
    Under fast fading, for any $\alpha>2$, there exists
    $b(n)=O(\log^6(n))$ such that for all $V(n)\in\mc{V}(n)$,
    \begin{align*}
        \Lambdamc(n) \subset b(n) \hLambdamc(n).
    \end{align*}
    The same statement holds under slow fading with probability
    $1-o(1)$ as $n\to\infty$.
\end{lemma}
\begin{proof}
    We say that a unicast traffic matrix $\lambdauc$ is
    \emph{compatible} with a multicast traffic matrix $\lambdamc$ if there
    exists a mapping $f:V(n)\times 2^{V(n)}\to V(n)$ such that $f(u,W)\in
    W\cup\{u\}$, for all $(u,W)$, and
    \begin{equation*}
        \lambdauc_{u,w} = \sum_{\substack{W\subset V(n): \\ f(u,W)=w}}
        \lambdamc_{u,W}
    \end{equation*}
    for all $(u,w)$. In words, $\lambdamc$ is compatible with
    $\lambdauc$ if we can create the unicast traffic matrix $\lambdauc$ from
    $\lambdamc$ by simply discarding the traffic for the pair
    $(u,W)$ at all the nodes $W\setminus\{f(u,W)\}$. 

    Note that if $\lambdamc\in \Lambdabmc$ and if $\lambdauc$ is
    compatible with $\lambdamc$ then $\lambdauc\in \Lambdauc$.  Indeed,
    we can reliably transmit at rate $\lambdauc$ by using the
    communication scheme for $\lambdamc$ and discarding all the
    unwanted messages delivered by this scheme. Now consider a cut
    $V_{\ell,i}$ with $\ell \leq L(n)$ in the wireless network, and
    choose a mapping $f:V(n)\times 2^{V(n)}\to V(n)$ such that
    \begin{equation*}
        \sum_{u\in V_{\ell,i}}\sum_{w\notin V_{\ell,i}} \lambdauc_{u,w}
        = \sum_{u\in V_{\ell,i}}
        \sum_{\substack{W\subset V:\\ W\setminus V_{\ell,i}\neq\emptyset}}\lambdamc_{u,W}.
    \end{equation*}
    Since $\lambdauc\in\Lambdauc$, we can apply Lemma \ref{thm:cutset}
    to obtain
    \begin{equation*}
        \sum_{u\in V_{\ell,i}}
        \sum_{\substack{W\subset V:\\ W\setminus V_{\ell,i}\neq\emptyset}}\lambdamc_{u,W}
        = \sum_{u\in V_{\ell,i}}\sum_{w\notin V_{\ell,i}} \lambdauc_{u,w}
        \leq b(n) \big(4^{-\ell}n\big)^{2-\min\{3,\alpha\}/2},
    \end{equation*}
    with $b(n)= O(\log^6(n))$. Repeating the same argument for cuts of
    the form $\{u\}$ and $\{u\}^c$ and using Lemma \ref{thm:cutset2},
    shows that $\lambdamc\in b(n)
    \hLambdamc$. Noting that the $b(n)$ term is uniform in $\lambdamc$
    yields that
    \begin{equation*}
        \Lambdamc \subset b(n)\hLambdamc,
    \end{equation*}
    concluding the proof of the lemma.
\end{proof}

\section{Discussion}
\label{sec:discussion}

We discuss several aspects and extensions of the three-layer
architecture introduced in Section \ref{sec:schemes_hr} and used in the
achievability parts of Theorems \ref{thm:unicast} and
\ref{thm:multicast}. In Section \ref{sec:tree}, we comment on the
various tree structures used in the three-layer architecture.  In
Section \ref{sec:second} we show that for certain values of $\alpha$ the
bounds in the theorems can be significantly sharpened. In Section
\ref{sec:nonbalanced}, we discuss bounds for traffic that is not
balanced. In Section \ref{sec:large}, we show that for large values of
path-loss exponent ($\alpha>5$) these bounds are tight. Hence in the
large path-loss regime the requirement of balanced traffic is not
necessary, and we obtain a scaling characterization of the entire
unicast and multicast capacity regions. In Section \ref{sec:dense}, we
point out how the results discussed so far can be used to obtain the
scaling of the unicast and multicast capacity regions of dense
networks (where $n$ nodes are randomly placed on a square of unit area). 

\subsection{Tree Structures}
\label{sec:tree}

There are two distinct tree structures that are used in the construction
of the three-layer communication scheme proposed in this paper---one
explicit and one implicit. These two tree structures appear in different
layers of the communication scheme and serve different purposes.

The first (explicit) tree structure is given by the tree $G$ utilized in
the routing layer and implemented in the cooperation layer. The main
purpose of this tree structure is to perform localized load balancing.
In fact, the distribution and concentration of traffic is used 
to avoid unnecessary bottlenecks. Note that the tree $G$ is used by the
scheme for any value of $\alpha$.

The second (implicit) tree structure occurs in the physical layer. This
tree structure appears only for $\alpha\in(2,3]$. In this regime, the
physical layer uses the hierarchical relaying scheme. It is the
hierarchical structure of this scheme that can equivalently be
understood as a tree. The purpose of this second tree structure is to
enable distributed multiple-antenna communication, i.e., to perform
cooperative communication.

\subsection{Second-Order Asymptotics}
\label{sec:second}

The scaling results in Theorems \ref{thm:unicast} and
\ref{thm:multicast} are up to a factor $n^{\pm o(1)}$ and hence preserve
information at scale $n^\beta$ for constant $\beta$ (see also the
discussion in Section \ref{sec:implications}). Here we examine in more
detail the behavior of this $n^{\pm o(1)}$ factor and show that in
certain situations it can be significantly sharpened.

Note first that the outer bound in Theorems \ref{thm:unicast} and
\ref{thm:multicast} hold up to a factor $O(\log^6(n))$, i.e.,
poly-logarithmic in $n$. However, the inner bound holds only up to the
aforementioned $n^{-o(1)}$ factor. A closer look at the proofs of the
two theorems reveals that the precise inner bound is of order
\begin{equation*}
    \gamma^{-1}(n)n^{-O\big(\log^{-1/3}(n)\big)},
\end{equation*}
where $\gamma(n)$ is the factor in the definition of $\Buc(n)$ and
$\Bmc(n)$ (see \eqref{eq:unicast_balanced} and
\eqref{eq:multicast_balanced}).  With a more careful analysis (see
\cite{nie} for the details), this can be sharpened to essentially
\begin{equation*}
    \gamma^{-1}(n)n^{-O\big(\log^{-1/2}(n)\big)}.
\end{equation*}

The exponent $\log^{-1/2}(n)$ in the inner bound has two causes. The
first is the use of hierarchical relaying (for $\alpha\in(2,3]$). The
second is the operation of the physical layer at level $L(n)+1$ of the
tree (i.e., to implement communication between the leaf nodes of $G$ and
their parents). Indeed at that level, we are operating on a square of
area
\begin{equation*}
    4^{-L(n)}n = n^{\log^{-1/2}(n)},
\end{equation*}
and the loss is essentially inversely proportional to that area.  Now,
the reason why $L(n)$ can not be chosen to be larger (to make this loss
smaller), is because hierarchical relaying requires a certain amount of
regularity in the node placement, which can only be guaranteed for large
enough areas.

This suggests that for the $\alpha>3$ regime, where multi-hop
communication is used at the physical layer instead of hierarchical
relaying, we might be able to significantly improve the inner bound.  To
this end, we have to choose more levels in the tree $G$, such that at
the last level before the tree nodes, we are operating on a square that
has an area of order $\log(n)$. Changing the three-layer architecture in
this manner, and choosing $\gamma(n)$ appropriately, for $\alpha>3$ the
inner bound can be improved to $\Omega(\log^{-2}(n))$ in $n$.  Combined
with the poly-logarithmic outer bound, this yields a $O(\log^8(n))$
approximation of the balanced unicast and multicast capacity regions for
$\alpha>3$.

\subsection{Non-Balanced Traffic}
\label{sec:nonbalanced}

Theorems \ref{thm:unicast} and \ref{thm:multicast} describe the scaling
of the balanced unicast and multicast capacity regions $\Lambdabuc(n)$
and $\Lambdabmc(n)$, respectively. As we have argued, the balanced
unicast region $\Lambdabuc(n)$ coincides with the unicast capacity
region $\Lambdauc(n)$ along at least $n^2-n$ out of $n^2$ total
dimensions, and the balanced multicast region $\Lambdabmc(n)$ coincides
with the multicast capacity region $\Lambdauc(n)$ along at least
$n2^n-n$ out of $n2^n$ total dimensions. However, the proofs of these
results provide also bounds for traffic that is not balanced, i.e., for
the remaining $2n$ dimensions.

Define the following two regions:
\begin{align*}
    \hLambdauc_1(n)
    \defeq \Big\{ \lambdauc\in\Rp^{n\times n}: &
    \sum_{u\in V_{\ell,i}(n)}\sum_{w\notin V_{\ell,i}(n)}(\lambdauc_{u,w}+\lambdauc_{w,u})
    \leq (4^{-\ell}n)^{2-\min\{3,\alpha\}/2} \\
    & \qquad\qquad \forall \ell\in\{1,\ldots,L(n)\}, i\in\{1,\ldots, 4^\ell\}, \\
    & \sum_{w\neq u}(\lambdauc_{u,w}+\lambdauc_{w,u}) \leq 1 
    \ \forall u\in V(n) 
    \Big\},
\end{align*}
and
\begin{align*}
    \hLambdamc_1(n) 
    \defeq \Big\{\lambdamc\in\Rp^{n\times 2^n}: 
    & \sum_{u\in V_{\ell,i}(n)} 
    \sum_{\substack{W\subset V(n):\\ W\setminus V_{\ell,i}(n)\neq\emptyset}}\lambdamc_{u,W}
    +\sum_{u\notin V_{\ell,i}(n)} 
    \sum_{\substack{W\subset V(n):\\ W\cap V_{\ell,i}(n)\neq\emptyset}}\lambdamc_{u,W} 
    \leq (4^{-\ell}n)^{2-\min\{3,\alpha\}/2} \\
    & \qquad\qquad \ \forall \ell\in\{1,\ldots,L(n)\}, i\in\{1,\ldots, 4^\ell\}, \\
    & \qquad 
    \sum_{\substack{W\subset V(n):\\ W\setminus \{u\}\neq\emptyset}}\lambdamc_{u,W}
    +\sum_{\tilde{u}\neq u}
    \sum_{\substack{W\subset V(n):\\ u\in W}}\lambdamc_{\tilde{u},W}
    \leq 1 \ \forall u\in V(n) 
    \Big\}.
\end{align*}
$\hLambdauc(n)$ and $\hLambdauc_1(n)$ differ in that for
$\ell\in\{1,\ldots,L(n)\}$, $\hLambdauc(n)$ only bounds traffic flow out
of $V_{\ell,i}(n)$, whereas $\hLambdauc_1(n)$ bounds traffic in both
directions across $V_{\ell,i}(n)$ (and similar for $\hLambdamc(n)$ and
$\hLambdamc_1(n)$). 

The analysis in Sections \ref{sec:proof_unicast} and
\ref{sec:proof_multicast} shows that
\begin{align*}
    n^{-o(1)}\hLambdauc_1(n) & \subset \Lambdauc(n) \subset O(\log^6(n))\hLambdauc(n), \\
    n^{-o(1)}\hLambdamc_1(n) & \subset \Lambdamc(n) \subset O(\log^6(n))\hLambdamc(n),
\end{align*}
with probability $1-o(1)$ as $n\to\infty$. In other words, we obtain an
inner and an outer bound on the capacity regions $\Lambdauc(n)$ and
$\Lambdamc(n)$. These bounds coincide in the scaling sense for balanced
traffic, for which we recover Theorems \ref{thm:unicast} and
\ref{thm:multicast}.

\subsection{Large Path-Loss Exponent Regime}
\label{sec:large}

The discussion in Section \ref{sec:nonbalanced} reveals that in order to
obtain scaling information for traffic that is not balanced, a stronger
version of the converse results in Lemma \ref{thm:cutset} is needed. In
particular, Lemma \ref{thm:cutset} bounds the sum-rate 
\begin{equation*}
    \sum_{u\in V_{\ell,i}(n)}\sum_{w\notin V_{\ell,i}(n)}\lambdauc_{u,w}
\end{equation*}
for $\lambdauc\in\Lambdauc(n)$. The required stronger version of the
lemma would also need to bound sum rates in the other direction, i.e., 
\begin{equation*}
    \sum_{u\notin V_{\ell,i}(n)}\sum_{w\in V_{\ell,i}(n)}\lambdauc_{u,w}.
\end{equation*}

For large path-loss exponents $\alpha>5$, such a stronger version of
Lemma \ref{thm:cutset} holds (see Lemma \ref{thm:cutset3}). With this,
we obtain that for $\alpha>5$,
\begin{align*}
    n^{-o(1)}\hLambdauc_1(n) & \subset \Lambdauc(n) \subset O(\log^6(n))\hLambdauc_1(n), \\
    n^{-o(1)}\hLambdamc_1(n) & \subset \Lambdamc(n) \subset O(\log^6(n))\hLambdamc_1(n),
\end{align*}
with probability $1-o(1)$ as $n\to\infty$. In other words, in the high
path-loss exponent regime $\alpha>5$, $\hLambdauc_1(n)$ and
$\hLambdamc_1(n)$ characterize the scaling of the entire unicast and
multicast capacity regions, respectively.

\subsection{Dense Networks}
\label{sec:dense}

So far, we have only discussed \emph{extended} networks, i.e., $n$ nodes
are located on a square of area $n$. We now briefly sketch how these
results can be recast for \emph{dense} networks, in which $n$ nodes are
located on a square of unit area. 

Note first that by rescaling power by a factor $n^{-\tilde{\alpha}/2}$,
a dense network with any path-loss exponent $\alpha$ can essentially be
transformed into an equivalent extended network with path-loss exponent
$\tilde{\alpha}$. In particular, any scheme for extended networks with
path-loss exponent $\tilde{\alpha}$ yields a scheme with same
performance for dense networks with any path-loss exponent $\alpha$ (see
also \cite[Section V.A]{ozg}). To optimize the resulting scheme for the
dense network, we start with the scheme for extended networks
corresponding to $\tilde{\alpha}$ close to $2$.  Hence an inner bound
for the unicast and multicast capacity regions for dense networks with
path-loss exponent $\alpha$ can be obtained from the ones for extended
networks by taking a limit as $\tilde{\alpha}\to 2$.  Moreover, an
application of Lemma \ref{thm:cutset2} yields a matching (in the scaling
sense) outer bound.

The resulting approximate balanced capacity regions $\hLambdabuc(n)$ and
$\hLambdabmc(n)$ have particularly simple shapes in this limit. In fact,
the only constraints in \eqref{eq:approx_unicast} and
\eqref{eq:approx_multicast} that can be tight are at level
$\ell=\log(n)$. Moreover, as in Section \ref{sec:large}, it can be shown
that the restriction of balanced traffic is not necessary for dense
networks. This results in the following approximate capacity regions for
dense networks:
\begin{equation*}
    \hLambdauc(n)
    \defeq \Big\{\lambdauc\in\Rp^{n\times n}:
    \sum_{w\neq u}(\lambdauc_{u,w}+\lambdauc_{w,u})
    \leq 1, \forall \ u\in V(n) \Big\}
\end{equation*}
for unicast, and
\begin{equation*}
    \hLambdamc(n) 
    \defeq \Big\{ \lambdamc\in\Rp^{n\times 2^n}:  
    \sum_{\substack{W\subset V(n):\\ W\setminus\{u\}\neq\emptyset}}\lambdamc_{u,W} 
    +\sum_{\tilde{u}\neq u}\sum_{\substack{W\subset V(n):\\ u\in W}}\lambdamc_{\tilde{u},W}
    \leq 1, \ \forall u\in V(n) \Big\}
\end{equation*}
for multicast. We obtain that for dense networks, for any $\alpha>2$,
\begin{align*}
    n^{-o(1)}\hLambdauc(n) & \subset \Lambdauc(n) \subset O(\log^6(n))\hLambdauc(n), \\
    n^{-o(1)}\hLambdamc(n) & \subset \Lambdamc(n) \subset O(\log^6(n))\hLambdamc(n),
\end{align*}
with probability $1-o(1)$ as $n\to\infty$.

\section{Conclusions}
\label{sec:conclusions}

In this paper, we have obtained an explicit information-theoretic
characterization of the scaling of the $n^2$-dimensional balanced
unicast and $n 2^n$-dimensional balanced multicast capacity
regions of a wireless network with $n$ randomly placed nodes and
assuming a Gaussian fading channel model. These regions span at least
$n^2-n$ and $n2^n-n$ dimensions of $\Rp^{n\times n}$ and $\Rp^{n\times
2^n}$, respectively, and hence determine the scaling of the unicast
capacity region along at least $n^2-n$ out of $n^2$ dimensions and the
scaling of the multicast capacity region along at least $n2^n-n$ out of
$n2^n$ dimensions. The characterization is in terms of $2n$ weighted
cuts, which are based on the geometry of the locations of the source
nodes and their destination nodes and on the traffic demands between
them, and thus can be readily evaluated. 

This characterization is obtained by establishing that the unicast and
multicast capacity regions of a capacitated (wireline, noiseless) tree
graph under routing have essentially the same scaling as that of the
original network. The leaf nodes of this tree graph correspond to the
nodes in the wireless network, and internal nodes of the tree graph
correspond to hierarchically growing sets of nodes.

This equivalence suggests a three-layer communication architecture for
achieving the entire balanced unicast and multicast capacity regions (in the
scaling sense). The top or routing layer establishes paths from each of
the source nodes to its destination (for unicast) or set of destinations
(for multicast) over the tree graph.  The middle or cooperation layer
provides this tree abstraction to the routing layer by distributing the
traffic among the corresponding set of nodes as a message travels up the
tree graph, and by concentrating the traffic on to the corresponding
set of nodes as a message travels down the tree. The bottom or
physical layer implements this distribution and concentration of traffic
over the wireless network. This implementation depends on the path-loss
exponent: For low path loss, $\alpha \in (2,3]$, hierarchical relaying
is used, while for high path loss ($\alpha>3$), multi-hop communication
is used. 

This scheme also establishes that a separation based approach, in which the
routing layer works essentially independently of the physical layer, can
achieve nearly the entire balanced unicast and multicast capacity
regions in the scaling sense. Thus, for balanced traffic, such
techniques as network coding can provide at most a small increase in the
scaling.

\section{Acknowledgments}

The authors would like to thank David Tse, Greg Wornell, and Lizhong
Zheng for helpful discussions, and the anonymous reviewers for their
help in improving the presentation of this paper.

\end{document}